\documentclass[journal]{IEEEtran}

\usepackage{hyperref,doi,url}
\hypersetup{colorlinks=true, linkcolor=blue, breaklinks=true, urlcolor=blue}

\usepackage[table]{xcolor}
\usepackage{graphicx}
\usepackage{amsmath}
\usepackage{amssymb}
\usepackage{amsthm}
\usepackage{bm}
\usepackage[compress]{cite}
\usepackage{enumitem}
\usepackage{dsfont}
\usepackage{bbm}
\usepackage{bbold}
\usepackage{indentfirst}
\usepackage{subfigure}
\usepackage{hhline}
\usepackage{mathtools}

\usepackage[prependcaption,colorinlistoftodos]{todonotes}

\usepackage{lineno}
\usepackage{multirow}
\usepackage{amsfonts}
\usepackage{textcomp}
\usepackage{stfloats}

\usepackage{tikz}
\usetikzlibrary{arrows.meta}
\tikzset{%
	>={Latex[width=2mm,length=2mm]},
	base/.style = {rectangle, rounded corners, draw=black,
		minimum width=2cm, minimum height=1cm,
		text centered, font=\sffamily},
	activityStarts/.style = {base, fill=blue!30},
	startstop/.style = {base, fill=red!30},
	activityRuns/.style = {base, fill=green!30},
	process/.style = {base, minimum width=2.5cm, fill=orange!15,
		font=\ttfamily},
}

\usepackage{centernot}
\usepackage{version,xspace}
\usepackage{environ}
\usepackage{blkarray}
\usetikzlibrary{automata,positioning,arrows,through}

\usepackage{threeparttable}

\usepackage{float}
\usepackage{indentfirst}
\usepackage{mathrsfs}
\usepackage[linesnumbered,ruled,vlined]{algorithm2e}

\usepackage{algorithmic}
\usepackage{tabularx}

\newtheorem{defi}{\textbf{Definition}}
\newtheorem{thom}[defi]{\textbf{Theorem}}
\newtheorem{asp}[defi]{\textbf{Assumption}}

\newtheorem{pro}[defi]{\textbf{Proposition}}
\newtheorem{lema}[defi]{\textbf{Lemma}}

\newtheorem{cor}[defi]{\textbf{Corollary}}

\newcommand{\defiref}[1]{Definition \ref{#1}}
\newcommand{\thomref}[1]{Theorem~\ref{#1}}
\newcommand{\aspref}[1]{Assumption~\ref{#1}}

\newcommand{\proref}[1]{Proposition \ref{#1}}

\newcommand{\lemaref}[1]{Lemma \ref{#1}}
\newcommand{\tabref}[1]{Table \ref{#1}}

\begin{document}

\title{Policy Evaluation and Seeking for Multi-Agent Reinforcement Learning
  via Best Response}

\author{Rui~Yan,~\IEEEmembership{Student Member,~IEEE,} Xiaoming~Duan,~Zongying~Shi,~\IEEEmembership{Member,~IEEE,}~Yisheng~Zhong,
	Jason R. Marden	and Francesco~Bullo,~\IEEEmembership{Fellow,~IEEE}
	\thanks{The work of R. Yan, Z. Shi, and Y. Zhong was supported in part by
		the National Natural Science Foundation of China under Grant 61374034,
		and in part by China Scholarship Council. This work of X. Duan and
		F. Bullo was supported in part by the US Air Force Office of Scientific
		Research under award FA9550-15-1-0138.}
	\thanks{R. Yan, Z. Shi, and Y. Zhong are with the Department
		of Automation, Tsinghua University, Beijing 100084, China. {\tt\small
			yr15 @mails.tsinghua.edu.cn} and {\tt\small\{szy,zys-dau\}
			@mail.tsinghua.edu.cn}}
	\thanks{X. Duan and F. Bullo are with the
		Mechanical Engineering Department and the Center of Control, Dynamical
		Systems and Computation, UC Santa Barbara, CA 93106-5070, USA. {\tt\small\{xiaomingduan,bullo\}@ucsb.edu}}
	\thanks{J. R. Marden is with the Department of Electrical and Computer Engineering, UC Santa Barbara, CA 93106-5070, USA.
	{\tt\small jrmarden@ece.ucsb.edu}}
}

% make the title area
\maketitle

\IEEEpeerreviewmaketitle
\begin{abstract}
  This paper introduces two metrics (cycle-based and memory-based metrics),
  grounded on a dynamical game-theoretic solution concept called \emph{sink
    equilibrium}, for the evaluation, ranking, and computation of policies
  in multi-agent learning. We adopt strict best response dynamics (SBRD) to
  model selfish behaviors at a meta-level for multi-agent reinforcement
  learning. Our approach can deal with dynamical cyclical behaviors (unlike
  approaches based on Nash equilibria and Elo ratings), and is more
  compatible with single-agent reinforcement learning than $\alpha$-rank
  which relies on weakly better responses.
  We first consider settings where the difference between largest and
  second largest underlying metric has a known lower bound.  With this
  knowledge we propose a class of perturbed SBRD with the following
  property: only policies with maximum metric are observed with nonzero
  probability for a broad class of stochastic games with finite memory.
  We then consider settings where the lower bound for the difference is
  unknown.  For this setting, we propose a class of perturbed SBRD such
  that the metrics of the policies observed with nonzero probability differ
  from the optimal by any given tolerance. The proposed perturbed SBRD
  addresses the opponent-induced non-stationarity by fixing the strategies
  of others for the learning agent, and uses empirical game-theoretic
  analysis to estimate payoffs for each strategy profile obtained due to
  the perturbation.
\end{abstract}

%\fbtodo{maybe title too long; try to keep at 10, 11 words}

\section{Introduction}
\emph{Motivation and problem description:} Multi-agent policy evaluation
and learning are long-standing challenges in multi-agent reinforcement
learning (MARL) since agents interact and learn in a complex environment
simultaneously through, competition such as in Go \cite{DS-AH-CM-AG-LS:16},
cooperation such as learning to communicate \cite{JF-IA-Nd-SW:16}, or some
mix of the two \cite{JL-VZ-ML-LG:17}. This paper introduces two metrics for
the evaluation and ranking of policies in multi-agent interactions,
grounded on a dynamical game-theoretic solution concept called \emph{sink
  equilibrium} first introduced by Goemans \emph{et al.} in
\cite{MG-VM-AV:05}. We aim at developing new learning paradigms with a
finite memory to achieve desirable convergence properties for certain
classes of stochastic games. More specifically, we formulate the
interactions of multiple agents by stochastic games and then analyze them
in meta-level as meta-games (i.e., normal-from games with unknown payoffs
\cite{PM-SO-MR-KT:20}). Instead of focusing on atomic decisions, meta games
abstract the underlying games and are centered around high-level agents'
interactions. Each strategy for an agent is a set of specific
hyperparameters or policy representations. Then, we propose two metrics for
the joint strategies based on the memory size of the learning system and
the notion of sink equilibria. We further design a class of MARL
algorithms, called perturbed strict best response dynamics (SBRD), such
that the frequently observed strategies in the learning system either have
the maximum underlying metrics or have metrics that are within any given
difference from the maximum, depending on the prior information about the
underlying games.

\emph{Literature review:} In single-agent RL, the long-term reward received
by the agent can be used to evaluate policies such as $Q$-learning
\cite{JCW-PD:92}. However, in MARL, the evaluation of policies becomes less
obvious as each agent aims to maximize its own reward and the reward itself
is also affected by the strategies of others. For multi-agent policy
evaluation, the Elo rating system \cite{AEE:78} is a predominant approach,
but it fails to handle intransitive relations between the learning agents
(e.g., the cyclic behavior in Rock-Paper-Scissors). On the other hand, the
Nash equilibrium in game theory is the most common tool for evaluating
strategies of non-cooperative agents. However, the recent work
\cite{SO-CP-GP:19} claims that there seems little hope of using the Nash
equilibrium in general large-scale games due to its limitations, such as
computational intractability, selection issues, and incompatibility with
dynamical systems. To address some of the issues above, Omidshafiei
\emph{et al.} \cite{SO-CP-GP:19} proposed an alternative evaluation
approach, called $\alpha$-rank, which leverages evolutionary game theory to
rank strategies in multi-agent games. Specifically, $\alpha$-rank defines
an irreducible Markov chain over joint strategies and constructs the
\emph{response graph} of the game. The ranking of the joint strategies is
obtained by computing the stationary distribution of the Markov chain.
Recent extensions of $\alpha$-rank can be found in
\cite{MR-SO-KT-JP:19,PM-SO-MR-KT:20}. Despite this progress, how to
evaluate multi-agent policies in a principled manner remains largely an
open question \cite{MR-SO-KT-JP:19}.

In traditional reinforcement learning, a single agent improves her strategy
based on the observations obtained by repeatedly interacting with the
stationary environment. In a simplest form of MARL, known as
\emph{independent RL} (InRL), each agent independently learns its own
strategy by directly using single-agent learning algorithms and treating
other agents as part of the environment \cite{MT:93}. In this case, the
environment is no longer stationary from the perspective of any individual
agent due to the independent and simultaneous learning. To tackle this
problem, people create a stationary environment for the learning agent by
fixing the strategies of other agents in previous works. For example, a
unified game-theoretic approach, known as \emph{policy-space response
  oracles} (PSRO), is introduced in \cite{ML-VZ-AG-AL-KT:17} and further
investigated in \cite{PM-SO-MR-KT:20}, where best responses to a
distribution over the policies of other agents are computed. In
\cite{GA-SY-17}, the authors present decentralized $Q$-learning algorithms
for the weakly acyclic stochastic games, where the environment keeps
stationary for the learning agent over each \emph{exploration phase}. The
above methodology in which a stationary environment is created mainly
focuses on the evolution or convergence of finite high-level strategies
rather than primitive actions. Thus, these games are abstracted as
normal-form games with unknown payoffs, also known as meta-games or
strategic-form games \cite{PM-SO-MR-KT:20,JS-PY-GA-JS-09}.

The idea of playing best responses to stationary environment highlighted
above has drawn significant attention in the area of learning in
strategic-form games. In the well-known fictitious play \cite{GWB:51}, at
each time step, each agent plays a best response to the empirical frequency
of the opponents' previous plays. In double oracle \cite{HBM-GJG-AB:03},
each agent plays a best response to the Nash equilibrium of its opponent
over all previous learned strategies. In adaptive play \cite{HPY:93}, each
agent can recall a finite number of previous strategies and at each time
step, each agent samples a fixed number of elements from its memory and
plays a best response to the sampled strategies. The best response dynamics
is a simple and natural local search method: an arbitrary agent is chosen
to improve its utility by deviating to its best strategy given the profile
of others \cite{BS-CE-SK-AR:18,RE-AF-MS-DW:13}.

Best-response based algorithms are guaranteed to converge to the set of
Nash equilibria in many games of interest, e.g., fictitious play and double
oracle in two-agent zero-sum games \cite{GWB:51,HBM-GJG-AB:03}, adaptive
play in potential games \cite{HPY:93} and best response dynamics in weakly
acyclic games \cite{HPY:04}. However, ensuring convergence to a Nash
equilibrium in multi-agent general-sum games is a great challenge
\cite{RE-AF-MS-DW:13}.  Moreover, although pure Nash equilibria (PNE) are
preferable to mixed Nash equilibria (MNE), (e.g., in meta-games
\cite{SO-CP-GP:19} and auctions \cite{MG-VM-AV:05}), PNEs may not exist,
e.g., in weighted congestion games \cite{MG-VM-AV:05}, generic games
\cite{ACC-DSL-AR-NRJ:13}, and multi-agent general-sum normal-form games.

Motivated by $\alpha$-rank \cite{SO-CP-GP:19}, we propose to use the
concept of sink equilibrium introduced in \cite{MG-VM-AV:05} to evaluate
policies in multi-agent settings. The sink equilibrium is a sink strongly
connected component (SSCC) in the \emph{strict best response graph}
associated with a game. The strict best response graph has a vertex set
induced by the set of pure strategy profiles, and its edge set contains
myopic strict best response by any individual agent. We propose two metrics
for strategy profiles based on sink equilibria: when a strategy belongs to
a sink equilibrium, it has a positive metric dependent on the property of
the sink equilibrium and the underlying metric; otherwise, it has zero
metric. Similar to $\alpha$-rank, sink equilibria include PNEs as a special
case and are guaranteed to exist \cite{MG-VM-AV:05}. Compared with a Nash
equilibrium, a sink equilibrium has several advantages: 1) it focuses on
pure strategy profiles which are preferable in meta-games; 2) it can deal
with dynamical and intransitive behaviors; 3) it is always finite. Most of
current works about sink equilibria focus on the computational complexity
and price of anarchy for various games
\cite{VSM-AS:09,AF-CHP:08,AB-MH-KL-AR:08}.

Like in Nash equilibrium selection problem, it is appealing to design
algorithms that converge to a sink equilibrium with some maximum underlying
metric. The strict best response dynamics, defined as a walk over the
strict best response graph, guarantees that for any initial strategy
profile, the convergence to a sink equilibrium is achieved
\cite{MG-VM-AV:05}. Similar to stochastic adaptive play \cite{HPY:93}, we
propose a class of perturbed strict best response dynamics such that the
strategies with the maximum metric are found with high probability. We
analyze the perturbed SBRD through its induced Markov chain and use
\emph{stochastic stability} as a solution concept, which is a popular tool
for studying stochastic learning dynamics in games
\cite{HPY:93,HB-JM-JS:19,JRM:17,HZ-ZM-CP-LP:19,GCC:19}. Similar to PSRO
\cite{ML-VZ-AG-AL-KT:17}, we use empirical game-theoretic analysis (EGTA)
to study strategies obtained due to the perturbation, through simulation in
complex games \cite{WW-RD-GT-JK:02}. In an empirical game, expected payoffs
for each strategy profile are estimated and recorded in an empirical payoff
table.

\emph{Contributions:} In this paper, we study policy evaluation and
learning algorithms for MARL. By adopting concepts from game theory, we
introduce two policy evaluation metrics with better properties than the Elo
rating system \cite{AEE:78}, Nash equilibria and $\alpha$-rank
\cite{SO-CP-GP:19}. By combining stochastic learning dynamics with
empirical game theory, we also propose a class of learning algorithms with
good convergence properties for MARL. The contributions of this paper are
as follows.
\begin{enumerate}
\item Based on sink equilibria induced by the strict best response graph,
  we introduce two policy evaluation metrics, namely the \emph{cycle-based
    metric} and the \emph{memory-based metric}, for strategy evaluation in
  multi-agent general-sum games.  Policies in the same sink equilibrium
  have the same underlying metric. We also establish the connections
  between the concepts of coarse correlated equilibria \cite{JRM:17}, PNEs
  and the sink equilibria.
  
\item Under the assumptions that the sink equilibrium with the maximum
  underlying metric is unique and the difference between the maximum and
  second largest metrics has a known lower bound, we propose a class of
  perturbed SBRD such that the policies with the maximum metric are
  observed frequently over a finite memory. Specifically, when the memory
  size is sufficiently large, we consider the cycle-based metric and give a
  lower bound of the memory size and a class of perturbation functions to
  guarantee convergence. When the memory size is predefined, we consider
  the memory-based metric and give a class of perturbation functions to
  guarantee convergence.

\item Furthermore, when the lower bound of the difference between the
  maximum and second largest metrics is unknown, we propose a class of
  perturbed SBRD such that the metrics of the policies observed with high
  probability are close to the maximum metric within any prescribed
  tolerance.
\end{enumerate}

\emph{Paper organization:} In Section \ref{PreSection}, we model the MARL
as a stochastic game and then at a high-level, reformulate it as a
meta-game by considering stationary deterministic policies. In Section
\ref{BRDSESection}, based on the sink equilibrium induced by strict best
response graph, two multi-agent policy evaluation metrics are
introduced. In Section \ref{PBRDSection}, we propose training algorithms to
seek the best policies according to both metrics, respectively. We conclude
the paper in Section~\ref{sec:conclusion}.

\emph{Notation:} For any finite set $S$, let $|S|$ denote its
cardinality. For any positive integer $z$, let $[z]$ denote the set of
positive integers smaller than or equal to $z$, i.e.,
$[z]=\{1,2,\cdots,z\}$. Let $\textup{rand}(S)$ denote an element uniformly
selected from a finite set $S$. Let $\mathbb{R}$, $\mathbb{R}_{\ge0}$,
$\mathbb{R}_{>0}$, $\mathbb{N}$ and $\mathbb{N}_{>0}$ be the set of reals,
nonnegative reals, positive reals, nonnegative integers and positive
integers, respectively.

\section{Preliminaries}\label{PreSection}
We present here preliminaries in game theory.

\subsection{Stochastic Games}\label{SectionSG}
Stochastic games have long been used in MARL to model interactions between agents in a shared, stationary environment for developing MARL algorithms \cite{MLL:94}. A Markov decision process (MDP) is a stochastic game with only one agent. A finite (discounted) stochastic game $G_1$ is a tuple $(N,\mathcal{X},\mathcal{A},P,\mathcal{R},\beta)$, where:
\begin{enumerate}
	\item $N$ is a finite set of $n$ agents;
	\item $\mathcal{X}$ is a finite set of states;
	\item $\mathcal{A}=\mathcal{A}^1\times\cdots\times\mathcal{A}^n$, where $\mathcal{A}^i$ is a finite set of actions available to agent $i$;
	\item $P:\mathcal{X}\times\mathcal{X}\times\mathcal{A}\mapsto[0,1]$ is the transition probability function, where $P(x'\,|\,x,a)$ is the probability of transitioning from state $x\in\mathcal{X}$ to state $x'\in\mathcal{X}$ under action profile $a\in\mathcal{A}$;
	\item $\mathcal{R}=\mathcal{R}^1\times\cdots\times\mathcal{R}^n$, where $\mathcal{R}^i:\mathcal{X}\times\mathcal{A}\mapsto\mathbb{R}$ is a real-valued immediate reward for agent $i$; 
	\item $\beta=\beta^1\times\cdots\times\beta^n$, where $\beta^i\in(0,1)$ is a discount factor for agent $i$. 
\end{enumerate}

Such a stochastic game induces a discrete-time controlled Markov process. At time $t\in\mathbb{N}$, the system is in state $x_t$, and each agent chooses an action $a^i_t$. Then, according to the transition probability $P(\cdot\,|\,x_t,a_t)$, the process randomly transitions to state $x_{t+1}$ and each agent receives an immediate payoff $\mathcal{R}^i(x_t,a_t)$, where $a_t:=(a^1_t,\dots,a^n_t)\in\mathcal{A}$ is the action profile at time $t$.

A policy (or strategy) for an agent is a rule of choosing an appropriate action at any time based on the agent's information. We will consider the stationary deterministic policies for all agents. A stationary deterministic policy $s^i$ for agent $i$ is a mapping from $\mathcal{X}$ to $\mathcal{A}^i$ \cite{GA-SY-17}, that is, $a_t^i=s^i(x_t)$, and the only information available to agent $i$ at time $t$ is the current state $x_t$. We denote the set of such policies by $\mathcal{S}^i$, which is finite with $|\mathcal{S}^i|=|\mathcal{X}|^{|\mathcal{A}^i|}$.

Let $\mathcal{S}=\mathcal{S}^1\times\cdots\times\mathcal{S}^n$ be the set of joint strategies, and $\mathcal{S}^{-i}=\times_{j\neq i}\mathcal{S}^j$ denotes the set of possible strategies of all agents except agent $i$. We also use the notation $s^{-i}\in\mathcal{S}^{-i}$ to refer to the joint strategy of all agents except agent $i$, and sometimes write the joint strategies as $s=(s^{-i},s^i)$ for any $s\in\mathcal{S}$ and $i\in N$.

Given any joint strategy $s\in\mathcal{S}$, each agent $i$ receives an expected long-term discounted reward
\begin{equation}
V_s^i(x)=\mathds{E}\Big(\sum_{t\in\mathbb{N}}(\beta^i)^t\mathcal{R}^i(x_t,a_t)\,|\,x_0=x\Big),\ \ \forall x\in\mathcal{X},
\end{equation}
where $a_t=(a_t^1,\dots,a_t^n)$ and $a_t^i\in\mathcal{A}^i$ is selected according to $s^i(x_t)$, for all $i\in N$ and all $t\in\mathbb{N}$. The function $V_s^i(x)$ is the value function of agent $i$ under the joint strategy $s$ with the initial state $x$. The objective of each agent $i$ is to find a policy $s^i\in\mathcal{S}^i$ that maximizes $V_s^i(x)$ for all $x\in\mathcal{X}$. If only agent $i$ is learning and the other agents' strategies are fixed, then the stochastic game $G_1$ reduces to an MDP and it is well-known that there exists a stationary deterministic policy for agent $i$ that achieve the maximum of $V_s^i(x)$ for all $x\in\mathcal{X}$ and any fixed $s^{-i}\in\mathcal{S}^{-i}$  \cite{OHL-JBL:96}.

\subsection{Meta Games} 
We are interested in analyzing interactions in $G_1$ at a higher meta-level. A meta game is a simplified model of complex interactions which focuses on meta-strategies (or styles of play) other than atomic actions \cite{KT-JP-ML-JL-TG:18,SO-CP-GP:19,WW-RD-GT-JK:02}. For example, meta-strategies in poker may correspond to "passive/aggressive" or "tight/loose" strategies.  The (finite) meta game $G_2$, induced by a stochastic game $G_1$ defined in Section \ref{SectionSG}, is a tuple $(N,\mathcal{S},J)$ such that:
\begin{enumerate}
	\item $N$ is a finite set of $n$ agents;
	\item $\mathcal{S}=\mathcal{S}^1\times\cdots\times\mathcal{S}^n$ is a finite joint policy space, where $\mathcal{S}^i$ is a finite set of stationary  deterministic policies available to agent $i$ with $|\mathcal{S}^i|=|\mathcal{X}|^{|\mathcal{A}^i|}$;
	\item $J=J^1\times\cdots\times J^n$, where $J^i:\mathcal{S}\mapsto\mathbb{R}$ is a payoff function of agent $i$ defined as the average of the sum of the value functions $V_s^i(x)$ over all $x\in\mathcal{X}$, i.e., for any $s\in\mathcal{S}$,
	\begin{equation*}
	J^i(s)=\frac{1}{|\mathcal{X}|}\sum_{x\in\mathcal{X}}V^i_s(x),
	\end{equation*}
	and $J(s)=(J^1(s),\dots,J^n(s))^\top$.
\end{enumerate}

The objective of each agent $i$ is to find a policy $s^i\in\mathcal{S}^i$
that maximizes $J^i(s)$. We focus on the meta game $G_2$ in this
paper. Since different agents may have different payoff functions and each
agent's payoff also depends on the strategies of the other agents, we here
adopt the notion of equilibrium to characterize those policies that are
\emph{person-by-person optimal} \cite{GA-SY-17}.

A standard approach in analyzing the performance of systems controlled by non-cooperative agents is to examine the Nash equilibria. Note that in the stochastic game $G_1$ or its induced meta game $G_2$, a strategy $s^i$ can be learned by a machine learning agent $i$ and the function $V_s^i(x)$ captures the expected reward of agent $i$ when playing against the others in some game domain. The PNE for the meta game $G_2$ is defined as follows.

\begin{defi}[Pure Nash equilibrium]
	A strategy profile $s^*\in\mathcal{S}$ is a pure Nash equilibrium (PNE) if $J^i(s^*)\ge J^i(s^i,s^{*-i})$ for all strategies $s^i\in\mathcal{S}^i$ and all agents $i\in N$.
\end{defi}
It is known that pure Nash equilibria do not exist in many games \cite{MG-VM-AV:05}. Specifically, for games where multiple agents learn simultaneously, there is usually no prior information about the structure of payoff functions. Thus, it is not guaranteed that a PNE exists in these games and the PNE may not be an appropriate solution concept to evaluate policies.

\section{Strict Best Response Dynamics and Sink Equilibrium}\label{BRDSESection}
In this section, we introduce the strict best response dynamics and the sink equilibrium, as a preparation for the evaluation of multi-agent policies in the next section. 
\subsection{Motivation}
The  progress made on reinforcement learning has opened the way for creating autonomous agents that can learn by interacting with surrounding unknown environments \cite{DS-AH-CM-AG-LS:16,ML-VZ-AG-AL-KT:17,JF-IA-Nd-SW:16,GA-SY-17}. For single-agent reinforcement learning, the interactions between the agent and the stationary environment are often modeled by an MDP, i.e., the same action of the agent from the same state yields the same (possibly stochastic) outcomes. This is a fundamental assumption for many well-known reinforcement learning algorithms, such as $Q$-learning \cite{RS-AB:18} and deep $Q$-network \cite{VM-KK-DS-AG-IA:13}. However, for multiple independent learning agents, the environment is not stationary from each individual agent's perspective and the effects of an action also depend on the actions of other agents \cite{MT:93}. As a result, good properties of many single-agent reinforcement learning algorithms do not hold in this case.

We deal with this non-stationarity in multi-agent learning by creating a stationary environment for each learning agent. Specifically, we only allow one agent to learn at each learning phrase and fix all the other agents' strategies, as in \cite{GA-SY-17,PM-SO-MR-KT:20,SO-CP-GP:19}. Under this framework, any single-agent reinforcement learning algorithm can be adopted. Moreover, by focusing on the meta-level (meta game $G_2$), we use a game-theoretic concept called \emph{strict best response dynamics} \cite{MG-VM-AV:05,BS-CE-SK-AR:18} to analyze the multi-agent learning process and further introduce the \emph{sink equilibrium} \cite{MG-VM-AV:05}.

\subsection{Strict Best Response Dynamics and Sink Equilibrium}
A strategy $s^{*i}\in\mathcal{S}^i$ is called a best response of agent $i$ to a strategy profile $s^{-i}\in\mathcal{S}^{-i}$ if
\begin{equation*}
J^i(s^{*i},s^{-i})=\max_{s^i\in\mathcal{S}^i}J^i(s^{i},s^{-i}).
\end{equation*}
Note that for any fixed $s^{-i}\in\mathcal{S}^{-i}$, agent $i$ solves a stationary MDP problem. Moreover, since $\mathcal{X}$ and $\mathcal{S}^i$ are both finite, agent $i$ always has at least one stationary deterministic best response to any $s^{-i}\in\mathcal{S}^{-i}$ \cite{OHL-JBL:96}. We denote the set of stationary deterministic best responses by
\begin{equation*}
B^i(s^{-i})=\big\{s^{*i}\in\mathcal{S}^i\,|\,J^i(s^{*i},s^{-i})=\max_{s^i\in\mathcal{S}^i}J^i(s^i,s^{-i})\big\},
\end{equation*}
 which is nonempty and finite.

\begin{figure}
	\centering
	%%21.3,13.4
	\subfigure{
	\includegraphics[width=35mm,height=16mm]{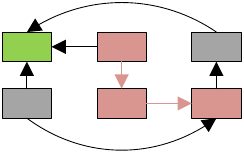}
	\put(-97,29.5){\scriptsize$a_1b_1$}
	\put(-57.5,29.5){\scriptsize$a_1b_2$}
	\put(-19,29.5){\scriptsize$a_1b_3$}
	\put(-97,12.5){\scriptsize$a_2b_1$}
	\put(-57.5,12.5){\scriptsize$a_2b_2$}
	\put(-19,12.5){\scriptsize$a_2b_3$}
	\put(-55,-7){\scriptsize$(a)$}
}
\subfigure{
		\includegraphics[width=43mm,height=14mm]{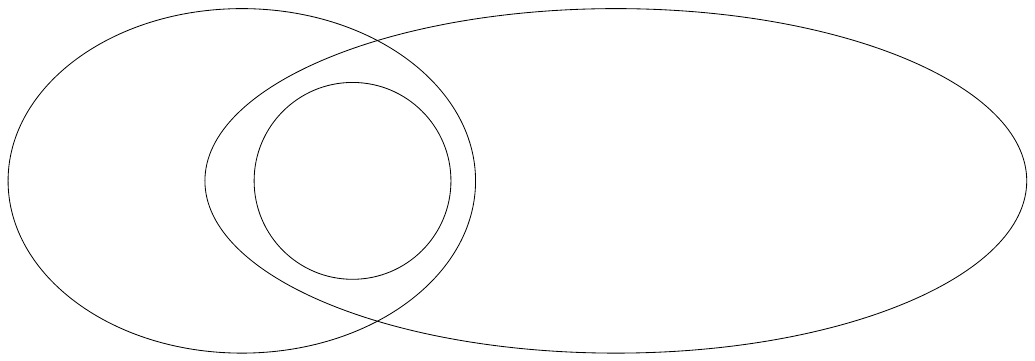}
		\put(-117,17){\footnotesize{CCE}}
		\put(-88,17){\footnotesize{PNE}}
		\put(-61,18){\footnotesize{sink equilibrium}}
		\put(-70,-7){\scriptsize$(b)$}
	}
	\caption{$(a)$ A strict best response graph for a meta game with two agents, where one of them has two strategies $\{a_1,a_2\}$ and the other has three strategies $\{b_1,b_2,b_3\}$. The sequence $((a_1,b_2),(a_2,b_2),(a_2,b_3))$ is a strict best response path (SBRP) and $\{(a_1,b_1)\}$ is the unique sink equilibrium. $(b)$ Relationships between different concepts of equilibrium.}
	\label{CCE_SE_figure}
\end{figure}
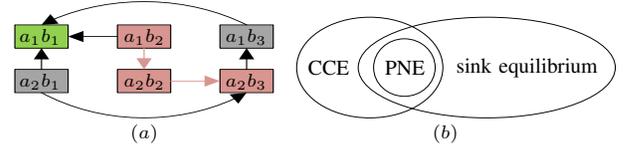

\begin{defi}[Strict best response graph]
	A strict best response graph $\mathcal{G}_{\textup{S}}=(\mathcal{S},\mathcal{E}_{\textup{S}})$ of a meta game $G_2$ is a digraph where each node represents a joint strategy profile $s\in\mathcal{S}$ and a directed edge $e_{ss'}$ from $s\in\mathcal{S}$ to $s'\in\mathcal{S}$ exists in $\mathcal{E}_{\textup{S}}$  if $s$ and $s'$ differ in exactly one agent's strategy (say $i\in N$) and $s'^{i}$ is a  best response with respect to $s^{-i}$, i.e., $s'^{i}\in B^i(s^{-i})$, and $J^i(s')>J^i(s)$.
\end{defi}

\begin{defi}[Strict best response path]
	A finite sequence $\mathcal{L}$ of joint strategies $(s_1,s_2,\dots,s_c)$ $(c\in\mathbb{N}_{>0})$ is called a strict best response path (SBRP) if, for each $\tau$\footnote{In the next section, we will formally define $\tau$ as exploration phrase.}, one of the two conditions holds: $1)$ $s_{\tau}$ is a PNE; $2)$ there is an edge from $s_{\tau}$ to $s_{\tau+1}$ in the strict best response graph $\mathcal{G}_{\textup{S}}$.
\end{defi}

For simplicity, we write $s\in \mathcal{L}$ when a joint strategy $s$ is on an SBRP $\mathcal{L}$. An SBRP $\mathcal{L}$ is a \emph{directed cycle}, if the only repeated nodes on $\mathcal{L}$ are the first and last nodes, or all nodes on $\mathcal{L}$ are the same PNE. An example of the strict best response graph is shown in Fig. \ref{CCE_SE_figure}(a), where the sequence in red is an SBRP. 

A \emph{strict best response dynamics (SBRD)}, also called strict Nash dynamics, is a walk on the strict best response graph \cite{MG-VM-AV:05}. In the rest of this paper, we stick with the random strict best response dynamics: at a given joint strategy, each agent is equally likely to be selected to play a strict best-response strategy \cite{HPY:04} if it exists. Moreover, the strict best response strategy of the agent is obtained by a single-agent reinforcement learning algorithm with all other agents' strategies fixed, and we assume that every strict best response strategy of the agent is selected with equal probability.  

Next, we introduce the concept of \emph{sink equilibrium} first proposed in \cite{MG-VM-AV:05} with a little modification, to examine the performance of a broad class of meta games, including those without pure Nash equilibria.

\begin{defi}[Sink equilibrium]\label{SinkDefi}
	A set $Q\subset\mathcal{S}$ of joint policies is a sink equilibrium of a strict best response graph $\mathcal{G}_{\textup{S}}=(\mathcal{S},\mathcal{E}_{\textup{S}})$, if $Q$ is a sink strongly connected component (SSCC) in the graph.
\end{defi}

Let the set of all sink equilibria be $\mathcal{Q}\subset2^{\mathcal{S}}$. All sink equilibria characterize all policies that are visited with non-zero probability after a sufficiently long random strict best-response sequence. Any random strict best-response sequence converges to a sink equilibrium with probability one \cite{MG-VM-AV:05}. In Fig. \ref{CCE_SE_figure}(a), $\{(a_1,b_1)\}$ is the unique sink equilibrium.   

\section{Sink Equilibrium Characterization and Multi-Agent Policy Evaluation}
In this section, we present our contribution on the properties of sink equilibrium and introduce two multi-agent policy evaluation metrics. 

\subsection{Sink equilibrium characterization}
The condensation digraph of any digraph is acyclic and every acyclic digraph has at least one sink \cite{FB:19}. Thus, it follows from \defiref{SinkDefi} that the sink equilibrium always exists in finite meta games, like (mixed) Nash equilibrium, i.e., $\mathcal{Q}$ is nonempty and finite. Moreover, sink equilibria generalize pure Nash equilibria in that a PNE is a singleton sink equilibrium. 

Next, we discuss the relationships between sink equilibrium and coarse
correlated equilibrium (CCE). The CCE, which has been broadly studied in
the area of learning in games \cite{JRM:17,HB-JM-JS:19}, is more general
than PNE as Fig. \ref{CCE_SE_figure}(b) shows. The CCE always exists in finite games, because the (mixed) Nash equilibrium is included in the CCE. Let
$q=\{q_s\}_{s\in\mathcal{S}}\in \Delta(\mathcal{S})$ be a probability
distribution over the joint strategy space $\mathcal{S}$, where
$\Delta(\mathcal{S})$ denotes the simplex over $\mathcal{S}$. A joint
strategy $s$ is in the \emph{support} of $q$ if $q_s>0$. In the next, when referring to the CCE, we focus on its support.

\begin{defi}[Coarse correlated equilibrium]\label{CorrelatedEqui}
	A probability distribution $q\in\Delta(\mathcal{S})$ is a coarse correlated equilibrium (CCE) if $\sum_{s\in\mathcal{S}}J^i(s)q_s\ge\sum_{s\in\mathcal{S}}J^i(s'^{i},s^{-i})q_s$ for all strategies $s'^{i}\in\mathcal{S}^i$ and all agents $i\in N$.
\end{defi}

%\fbtodo{the start of this proposition seems to imply existence and uniqueness of sink equilibria and CCEs}

The following proposition establishes the connections between CCE, PNE and the sink equilibrium.
\begin{pro}[Connection with the CCE]\label{ConnectionCCE}
	In finite meta games, sink equilibria and coarse correlated equilibria satisfy:
	\begin{enumerate}[label=(\roman*)]
		\item\label{itm:CCE1} 	There exist meta games such that any sink equilibrium is not equal to the support of any CCE;
		\item\label{itm:CCE2} There exist meta games such that the support of a CCE is not a subset of any sink equilibrium;
		\item The\label{itm:CCE3}re exist meta games such that the PNE is a proper subset of the intersection of the support of a CCE and a sink equilibrium.
	\end{enumerate}
\end{pro}

\begin{proof}
	Regarding~\ref{itm:CCE1}, we construct a meta game as in Fig. \ref{PayoffMatrix} with $0<\epsilon<1$, where row and column agents aim to maximize their payoffs. We draw the associated strict best response graph on the right according to the payoff matrix. This game has a unique sink equilibrium (four joint strategies in green):
	\begin{equation*}
	Q=\big\{(a_1,b_2),(a_2,b_2),(a_2,b_3),(a_1,b_3)\big\}.
	\end{equation*}
	
	Suppose that there is a CCE $q$ such that the sink equilibrium $Q$ is equal to the support of $q$, i.e., $q_s>0$ if and only if $s\in Q$. Let
	\begin{equation*}
	q_{(a_1,b_2)}=y_1,\,q_{(a_2,b_2)}=y_2,\,q_{(a_2,b_3)}=y_3,\,q_{(a_1,b_3)}=y_4,
	\end{equation*}
	where $y_i>0$ for $i\in\{1,2,3,4\}$ and $\sum_{i=1}^4y_i=1$.
	We take $s'^1=a_3$ and $s'^2=b_1$ in \defiref{CorrelatedEqui}. If $q$ is a CCE, then we have
	\begin{equation}\begin{aligned}\label{CorrelatedEqu}
	&y_2+y_4\ge (1-\epsilon)(y_1+y_2+y_3+y_4)=(1-\epsilon),\\
	&y_1+y_3\ge (1-\epsilon)(y_1+y_2+y_3+y_4)=(1-\epsilon),
	\end{aligned}
	\end{equation}
	which leads to a contradiction $\sum_{i=1}^4y_i>1$ for $0<\epsilon<\frac{1}{2}$. Therefore, the sink equilibrium is not equal to the support of CCE.
	
	Regarding \ref{itm:CCE2}, consider a meta game whose payoff matrix is a part of the payoff matrix in Fig.~\ref{PayoffMatrix}, where player one has strategies $\{a_1,a_3\}$ and player two has strategies $\{b_1,b_2\}$. Take $\epsilon=\frac{1}{3}$ and then $(0,1)$ for the row agent and $(\frac{2}{5},\frac{3}{5})$ for the column agent form an MNE which is support by $\{(a_3,b_1),(a_3,b_2)\}$. Note that this meta game only has a unique sink equilibrium $
	\{(a_3,b_2)\}$. Thus, the conclusion follows from the fact that the MNE is contained in the CCE \cite{JRM:17}. 

	Regarding \ref{itm:CCE3}, consider a meta game whose payoff matrix is a part of  the payoff matrix in Fig. \ref{PayoffMatrix} where only four strategies $a_1,a_2,b_2$ and $b_3$ are involved. It can be verified that $(\frac{1}{4},\frac{1}{4},\frac{1}{4},\frac{1}{4})$ is a CCE. Since these four joint strategies form a sink equilibrium, then we can claim that the PNE is a proper subset of the intersection set of sink equilibrium and CCE.
\end{proof}

The sink equilibrium and CCE are not special cases of each other and the PNE is a proper subset of their intersection set from \proref{ConnectionCCE}. The relationship between sink equilibrium and CCE is shown in Fig. \ref{CCE_SE_figure}(b). Thus, the learning algorithms in games for seeking the CCE cannot be directly applied for seeking sink equilibria with good properties. 

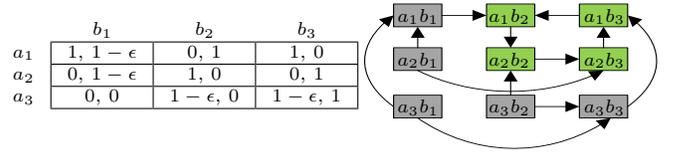
\begin{figure}
	\scriptsize
	\begin{tabular}{c|c|c|c|}
		\multicolumn{1}{c}{}& \multicolumn{1}{c}{$b_1$} & \multicolumn{1}{c}{$b_2$} & \multicolumn{1}{c}{$b_3$}\\
		\hhline{~|-|-|-}
		$a_1$&$1,\,1-\epsilon$ & 
		$0,\,1$ &  
		$1,\,0$\\
		\hhline{~|-|-|-}
		% \hhline{~|-|-|}
		
		$a_2$&$0,\,1-\epsilon$ &%\cellcolor[HTML]{CCFFCC} 
		$1,\,0$ &%\cellcolor[HTML]{CCFFCC}
		$0,\,1$\\ \hhline{~|-|-|-}
		$a_3$&$0,\,0$ & $1-\epsilon,\,0$ & $1-\epsilon,\,1$
		\\
		\hhline{~|-|-|-}
	\end{tabular}
	\put(-2,-31){\subfigure{
			\includegraphics[width=40mm,height=20mm]{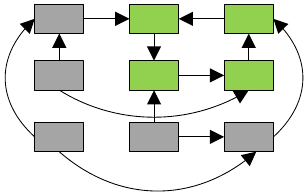}
			\put(-29,14.5){\scriptsize$a_3b_3$}
			\put(-65,14.5){\scriptsize$a_3b_2$}
			\put(-99.5,14.5){\scriptsize$a_3b_1$}
			\put(-29,32.5){\scriptsize$a_2b_3$}
			\put(-65,32.5){\scriptsize$a_2b_2$}
			\put(-99.5,32.5){\scriptsize$a_2b_1$}
			\put(-99.5,49){\scriptsize$a_1b_1$}
			\put(-65,49){\scriptsize$a_1b_2$}
			\put(-29,49){\scriptsize$a_1b_3$}
	}}
	\caption{A meta game with two agents, where the left is the payoff matrix and the right is the associated strict best response graph. The four joint strategies in green form a unique sink equilibrium of the left graph.}
	\label{PayoffMatrix}
\end{figure}

\subsection{Multi-agent policy evaluation}
\label{MEvaluationSec}
The problem of multi-agent policy evaluation is challenging due to several reasons: strategy and action spaces of agents quickly explode (e.g., multi-robot systems \cite{RY-ZS-YZ:20-2,RY-ZS-YZ:19}), models need to deal with intransitive behaviors (e.g. cyclic best-responses in Rock-Paper-Scissors, but at a much higher dimension \cite{SO-CP-GP:19}), types of interactions between agents may be complex (e.g., MuJoCo soccer), and payoffs for agents may be general-sum and asymmetric. 

Motivated by the recently-introduced $\alpha$-rank
strategy\cite{SO-CP-GP:19}, which essentially defines a walk on a perturbed
weakly better response graph, we use a walk on the strict best response
graph, i.e., SBRD, to evaluate the policies in multi-agent
settings. Similar to Nash equilibria, even simple games may have multiple
sink equilibria. Thus, we need to solve the sink equilibrium selection
problem.

We introduce two evaluation metrics for joint strategies and sink
equilibria in the following. Define the performance of a joint strategy
$s\in\mathcal{S}$ by
\begin{equation}\label{JointMeasure}
W(s)=\sum_{i\in N}w^iJ^i(s),
\end{equation}
where $w^i\in\mathbb{R}_{\ge0}$ is the weight associated with agent $i$ and $\sum_{i=1}^nw^i=1$. If we only care about the performance of agent $i$, then we can take $w^i=1$ and $w^j=0$ for all $j\neq i$. If we treat every agent equally and care about the average performance, then we can take $w^i=\frac{1}{n}$ for all $i\in N$. The performance of an SBRP $\mathcal{L}$ is defined as the average performance of all joint strategies in $\mathcal{L}$, i.e.,
\begin{equation}\label{PathMeasure}
W(\mathcal{L})=\frac{1}{|\mathcal{L}|}\sum_{s\in \mathcal{L}}W(s).
\end{equation}

We first introduce a cycle-based metric for policy evaluation. For a sink equilibrium $Q\in\mathcal{Q}$, let $\mathcal{C}(Q)$ be the set of directed cycles in the subgraph induced by $Q$. Moreover, if $Q=\{s\}$ is a singleton, then $\mathcal{C}(Q)=\{s\}$.

%\fbtodo{I provide an alternative definition here. See if you like it. If you do,
%  then change also the subsequent Definition 8}

\begin{defi}[Cycle-based metric]\label{CBMeasure}
The cycle-based metric $M_c(s)$ of a strategy
  $s\in\mathcal{S}$ is the worst performance of all directed cycles in
  $\mathcal{C}(Q)$ if $s\in Q$ for some sink equilibrium $Q\in\mathcal{Q}$
  and $0$ otherwise. In other words, 
  \begin{equation}\label{CycleMeasureEqu}
    M_c(s)=
    \begin{cases}
      \min_{\mathcal{L}\in\mathcal{C}(Q)}W(\mathcal{L}) & 
      \text{if $s\in Q$ for some $Q\in\mathcal{Q}$;}\\
      0 & otherwise.
    \end{cases}
  \end{equation}
  Furthermore, the cycle-based metric for a sink equilibrium
  $Q\in\mathcal{Q}$ is defined by $M_c(Q)=M_c(s)$ for any $s\in Q$.
\end{defi}

We next introduce a memory-based metric for policy evaluation in the case when we can only store $m\in\mathbb{N}_{>0}$ joint strategies. For a sink equilibrium $Q\in\mathcal{Q}$, let $\mathcal{M}(Q)$ be the set of SBRPs of length $m$ in the subgraph induced by $Q$.

\begin{defi}[Memory-based metric]\label{MBMeasure}
  Let $m\in \mathbb{N}_{>0}$ be the memory length. The memory-based metric $M_m(s)$ of a strategy
  $s\in\mathcal{S}$ is the worst performance of all SBRPs of length $m$ in
  $\mathcal{M}(Q)$ if $s\in Q$ for some sink equilibrium $Q\in\mathcal{Q}$
  and $0$ otherwise. In other words, 
	\begin{equation}\label{MemoryMeasureEqu}
	M_m(s)=
	\begin{cases}
	\min_{\mathcal{L}\in\mathcal{M}(Q)}W(\mathcal{L})& \text{if $s\in Q$ for some $Q\in\mathcal{Q}$;}\\
	0& otherwise.
	\end{cases}
	\end{equation}
	Furthermore, the memory-based metric for a sink equilibrium $Q\in\mathcal{Q}$ is defined by $M_m(Q)=M_m(s)$ for any $s\in Q$.
\end{defi} 

The reason that we consider the worst performance for each sink equilibrium in Definitions \ref{CBMeasure} and \ref{MBMeasure} is that the random strict best response dynamics can move along any SBRP in a sink equilibrium and the metrics give lower bound on the achievable performance. Since $\mathcal{Q}$ is finite and nonempty, we can rank all sink equilibria $Q\in\mathcal{Q}$ by the metrics $M_c(Q)$ or $M_m(Q)$ and would like the SBRD to converge to the sink equilibrium with the best performance.

\section{Multi-Agent Policy Seeking}\label{PBRDSection}
We next present our contribution on policy seeking. Since we have introduced two metrics to evaluate multi-agent policies, we propose training algorithms to seek the best policies according to both metrics, respectively. 

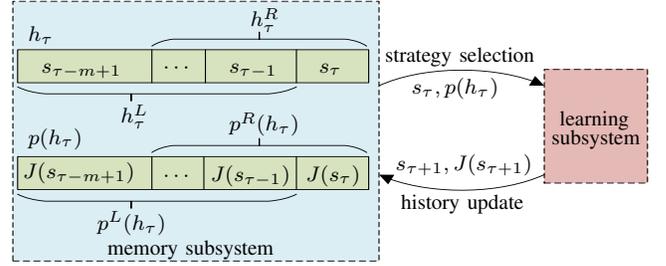
\begin{figure}
	\centering
	%%21.3,13.4
	\subfigure{
		\includegraphics[width=85mm,height=35mm]{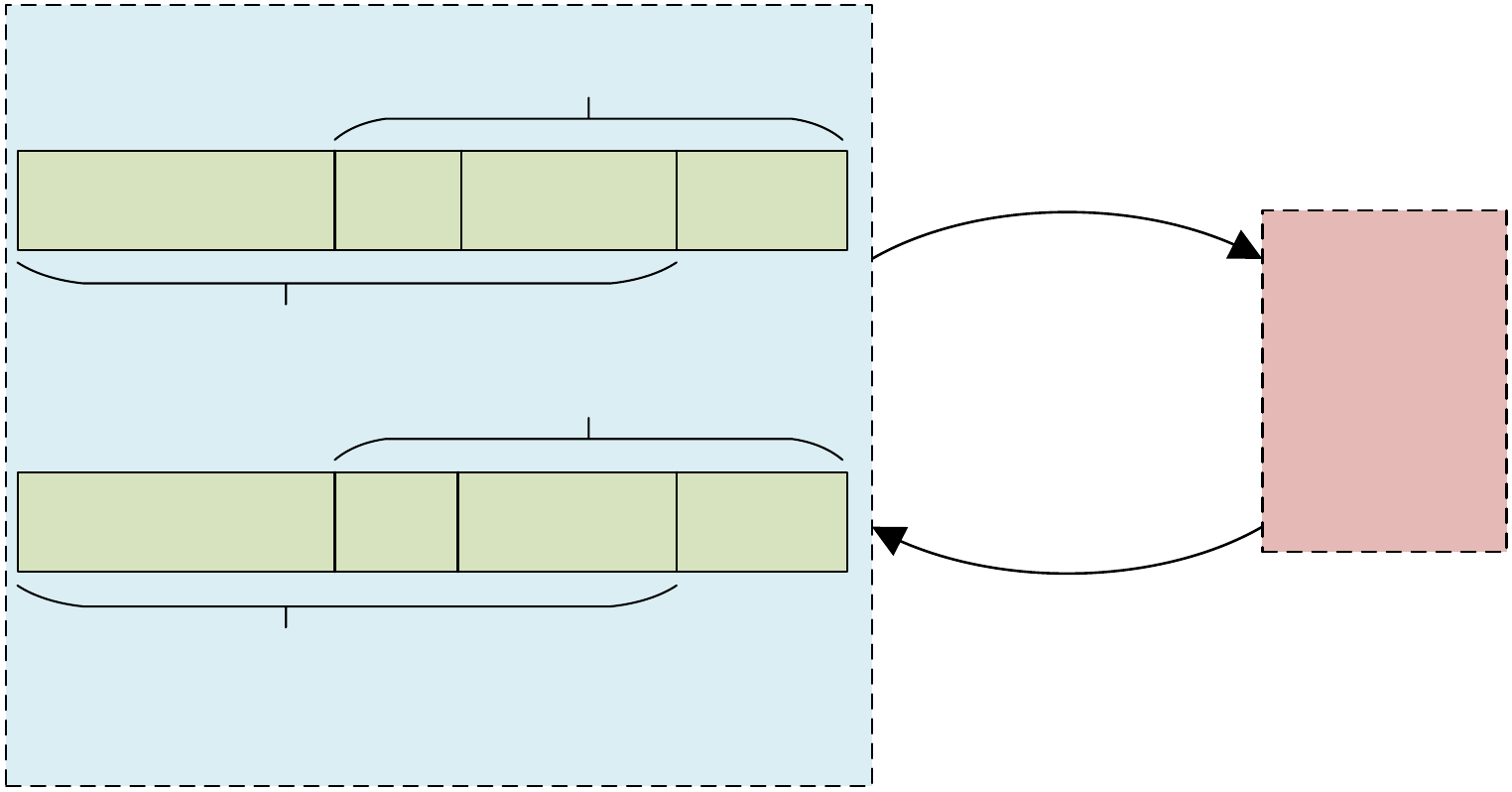}
		\put(-235,84){\footnotesize $h_{\tau}$}
		\put(-200,53){\footnotesize $h^L_{\tau}$}
		\put(-151,89.5){\footnotesize $h_{\tau}^R$}
		\put(-235,45){\footnotesize $p(h_{\tau})$}
		\put(-230,72){\footnotesize $s_{\tau-m+1}$}
		\put(-125,72){\footnotesize  $s_{\tau}$}
		\put(-237,32){\footnotesize  $J(s_{\tau-m+1})$}
		\put(-184,72){\footnotesize  $\cdots$}
		\put(-184,31){\footnotesize  $\cdots$}
		\put(-159,72){\footnotesize  $s_{\tau-1}$}
		\put(-159,50){\footnotesize  $p^R(h_{\tau})$}
		\put(-209,12.5){\footnotesize  $p^L(h_{\tau})$}
		\put(-166,31){\footnotesize  $J(s_{\tau-1})$}
		\put(-130,31){\footnotesize  $J(s_{\tau})$}
		\put(-205,3){\footnotesize{memory subsystem}}
		\put(-100,76){\footnotesize{strategy selection}}
		\put(-90,64){\footnotesize  $s_{\tau},p(h_{\tau})$}
		\put(-96,35){\footnotesize  $s_{\tau+1},J(s_{\tau+1})$}
		\put(-94,20.5){\footnotesize{history update}}
		\put(-34,52){\footnotesize{learning}}
		\put(-37,45){\footnotesize{subsystem}}
	}
	\caption{Sketch of the training system.}
	\label{training_system_figure}
\end{figure}

\subsection{Training System}
We first briefly introduce the framework of our training system which consists of a learning subsystem and a memory subsystem, shown in Fig. \ref{training_system_figure}. The learning subsystem is a training environment where multiple agents are learning and empirical games are simulated \cite{WW-RD-GT-JK:02}. In an empirical game, expected payoffs for each complex strategy profile are estimated and recorded in an empirical payoff table. The memory subsystem is a finite memory unit for storing the learned joint strategies.   
	
The meta game $G_2$ is played once in the learning subsystem during each period $\tau\in\{1,2,\dots\}$, where $\tau$ is the \emph{exploration phase} \cite{GA-SY-17}. At the exploration phrase $\tau$, the learning subsystem fetches the latest joint strategy $s_{\tau}$ and the payoff matrix $p(h_{\tau})$ (defined later) from the memory subsystem and combines them with reinforcement learning and empirical games to generate a new joint strategy $s_{\tau+1}$ and the related payoff $J(s_{\tau+1})$. Then, the memory subsystem receives $s_{\tau+1}$ and $J(s_{\tau+1})$, and uses a pruned history update rule to decide whether to store them. The process iterates afterwards.   

Our goal is to design algorithms such that after a long training period, the memory subsystem stores the joint strategies that have the maximum underlying metrics with probability one, regardless of the initial joint strategies.  

\subsection{Perturbed Strict Best Response Dynamics}
The proposed perturbed SBRD is described in \tabref{PSBRD}. Suppose that the memory subsystem possesses a finite memory of length $m$, and can recall the history of the previous $m$ joint strategies and the associated payoffs. Let $h_{\tau}$ and $p(h_{\tau})$ be the histories of joint strategies (called \emph{history state}) and payoffs up to exploration phrase $\tau$ in the memory subsystem, where $h_{\tau}=(s_{\tau-m+1},\dots,s_{\tau})$ and $p(h_{\tau})=(J(s_{\tau-m+1}),\dots,J(s_{\tau}))\in\mathbb{R}^{n\times m}$. Let $\mathcal{H}=\mathcal{S}^m$ denote the space of history states, and for any $h\in\mathcal{H}$, let $h^L$ and $h^{R}$ be the leftmost and rightmost $m-1$ joint strategies of $h$, respectively. Similar notations $p^L(h)$ and $p^R(h)$ are used for the payoffs. For example, for $h_{\tau}$ and $p(h_{\tau})$, we have that $h_{\tau}^L=(s_{\tau-m+1},\dots,s_{\tau-1})$, $h_{\tau}^R=(s_{\tau-m+2},\dots,s_{\tau}),p^L(h_{\tau})=(J(s_{\tau-m+1}),\dots,J(s_{\tau-1}))$ and $p^R(h_{\tau})=(J(s_{\tau-m+2}),\dots,J(s_{\tau}))$, as Fig. \ref{training_system_figure} illustrates. We here emphasize that $h_{\tau}$ can also be treated as a sequence of joint strategies $(s_{\tau-m+1},s_{\tau-m+2},\dots,s_{\tau})$ from $s_{\tau-m+1}$ to $s_{\tau}$.

The training system is initialized by simulating $G_2$ with a random selected joint strategy $s_1$ and then storing $s_1$ and the associated payoff vector $J(s_1)$ into the rightmost of memory subsystem. In the unperturbed process ($\epsilon=0$), only steps {\ref{item:SS}}, {\ref{item:RL}} and {\ref{item:HU}} of \tabref{PSBRD} are executed with $\bar{\epsilon}=0$. In these three steps, we update the history state under two conditions: if $s_{\tau}$ in current history state $h_{\tau}$ is a PNE, then we store this PNE $s_{\tau}$  regardless of the new joint strategy $s_{\tau+1}$; if $s_{\tau}$ is not a PNE and $(s_{\tau},s_{\tau+1})$ forms an SBRP, then we store $s_{\tau+1}$. When we store $s_{\tau+1}$ into the memory subsystem, the history state moves from $h_{\tau}$ to $h_{\tau+1}$ by removing the leftmost element $s_{\tau-m+1}$ of $h_{\tau}$ and appending $s_{\tau+1}$ as the rightmost element of $h_{\tau+1}$, i.e., $h_{\tau+1}=(h_{\tau}^R,s_{\tau+1})$. Similarly, when we store $s_{\tau}$, then $h_{\tau+1}=(h_{\tau}^R,s_{\tau})$. Otherwise, $h_{\tau+1}=h_{\tau}$. Similar operations are performed to update the history of the associated payoffs $p(h_{\tau+1})$.

For the perturbed process ($\epsilon>0$), we compute an exploration rate $\bar{\epsilon}$ as a function of the current payoff matrix $p(h_{\tau})$ in steps {\ref{item:Eval}} and {\ref{item:Exp}} of \tabref{PSBRD}, where the feasible function $f$ is a mapping from $\mathbb{R}^{n\times m}$ to $\mathbb{R}_{>0}$ whose form will be designed in \defiref{FeasibleFunc}. In this process, the strategy selection is slightly perturbed by $\bar{\epsilon}$ such that, with a small probability $\bar{\epsilon}$ agent $i$ follows a uniform strategy (or, it explores). The step {\ref{item:EG}} of \tabref{PSBRD} is executed when at least one agent explores, in which case we need to run an empirical game \cite{WW-RD-GT-JK:02} to obtain the payoff vector $J(s_{\tau+1})$ corresponding to $s_{\tau+1}$ as the value of payoff function $J^i(s)$ is unknown for all joint strategies $s\in\mathcal{S}$ and all agents $i\in N$ in advance. The history update is also perturbed by $\bar{\epsilon}$ such that with a small probability $\bar{\epsilon}$ the learned strategy $s_{\tau+1}$ and the related payoff $J(s_{\tau+1})$ are directly recorded.

We assume that in step {\ref{item:RL}} of \tabref{PSBRD}, the selected
agent is able to learn a best-response strategy. This can be achieved by
using some single-agent reinforcement learning algorithms (for example,
$Q$-learning \cite{RS-AB:18}), as long as the exploration phrase runs for a
sufficiently long time \cite{GA-SY-17}. Note that in the first $m$ plays of
the game, the memory is not full. We consider a sufficiently long
exploration phrase $\tau$ such that $\tau>m$. Thus, after a long run, i.e.,
$\tau>m$, the memory is always full and the history state satisfies
$h_{\tau}\in\mathcal{H}=\mathcal{S}^m$. Unless otherwise specified, we have
$\tau>m$ by default for all the proofs.

For the analysis below, we assume that the bound for each agent's payoff is known. 

\begin{asp}[Payoff bound]\label{PayoffBound}
	Suppose that the agents' payoff functions are non-negative and bounded by $J_{\rm max}>0$ from the above, i.e., $0\leq J^i(s)\leq J_{\rm max}$ for all strategies $s\in\mathcal{S}$ and all agents $i\in N$.
\end{asp}

\begin{table}
	\begin{center}
		\caption{Perturbed Strict Best Response Dynamics}
		\begin{tabular}{|p{8.1cm}|}
			\hline\vspace{0.01ex}
			\textbf{Initialize}: Take $p(h_1)=0_{n\times m}$. Randomly select a strategy $s_{1}^i$ from $\mathcal{S}^i$ for all $i\in N$. Then simulate $G_2$ with $s_1=(s^1_1,\dots,s^n_1)$ and $x_0=x$ for all $x\in\mathcal{X}$, and obtain the payoff vector $J(s_1)=(J^1(s_1),\dots,J^n(s_1))^{\top}$. Store $s_1$ and $J(s_1)$ into the system memory: $h_1=(h_1^R,s_1),p(h_1)=(p^R(h_1),J(s_1))$.
			
			\textbf{Learning process}: At time-steps $\tau=1,2,3,\dots,$
			\begin{enumerate}[ref=\arabic*)]
				\item
			    \label{item:Eval}	(\textbf{evaluation}) Use a feasible function $f:\mathbb{R}^{n\times m}\to\mathbb{R}_{>0}$ (defined in \defiref{FeasibleFunc}) to evaluate all the joint strategies in $h_{\tau}$ by their payoff matrix $p(h_{\tau})$: $\kappa=f(p(h_{\tau}))$; 
				\item\label{item:Exp}
				(\textbf{exploration}) Compute a payoff-based exploration rate: $\bar{\epsilon}=\epsilon^{\kappa}$;
				\item\label{item:SS} (\textbf{strategy selection}) Select an agent from $N$ randomly, say $j$. Agent $i\in N$ selects a new strategy as follows:
				\[
				s^i_{\tau+1}=
				\begin{cases}
				\textup{a best response}& \textup{with probability } 1-\bar{\epsilon} \textup{ if }i=j,\\
				\textup{rand}(\mathcal{S}^i)& \textup{with probability }\bar{\epsilon}\textup{ if }i=j,\\
				s^i_{\tau} &\textup{with probability }1-\bar{\epsilon} \textup{ if }i\neq j,\\
				\textup{rand}(\mathcal{S}^i) & \textup{with probability }\bar{\epsilon}\textup{ if }i\neq j;
				\end{cases}
				\]
				\item \label{item:RL} (\textbf{reinforcement learning}) If $s^j_{\tau+1}$ needs to be a best response, i.e., $s^j_{\tau+1}\in B^j(s^{-j}_{\tau})$, then
				
				\textbf{for} \emph{many episodes} \textbf{do}
				
				\quad Train $s^j_{\tau+1}$ over $s^{-j}_{\tau}$ with $x_0=x$ for all $x\in\mathcal{X}$;
				
				Obtain $J(s^j_{\tau+1},s^{-j}_{\tau})$. Then, take $J(s_{\tau+1})=J(s^j_{\tau+1},s^{-j}_{\tau})$;
				
				\item\label{item:EG} (\textbf{empirical game}) If there exists an $i\in N$ such that $s^i_{\tau+1}$ is obtained by $\text{rand}(\mathcal{S}^i)$ in step 3) (that is, agent $i$ explores), then simulate the game $G_2$ with the strategy profile $s_{\tau+1}$, and obtain the payoff vector $J(s_{\tau+1})$;

				\item\label{item:HU} (\textbf{history update}) With probability $1-\bar{\epsilon}$, the history follows the update rule: if $s_{\tau}$ is a PNE, then
				\begin{equation*}\begin{aligned}
				&h_{\tau+1}=(h_{\tau}^{R},s_{\tau}),\quad p(h_{\tau+1})=(p^R(h_{\tau}),J(s_{\tau}));
				\end{aligned}
				\end{equation*}
				if $s_{\tau}$ is not a PNE and $J^j(s_{\tau+1})>J^j(s_{\tau})$, then
				\begin{equation*}\begin{aligned}
				&h_{\tau+1}=(h^R_{\tau},s_{\tau+1}),\quad p(h_{\tau+1})=(p^R(h_{\tau}),J(s_{\tau+1}));
				\end{aligned}
				\end{equation*}
				otherwise,
				\begin{equation*}	h_{\tau+1}=h_{\tau},\quad p(h_{\tau+1})=p(h_{\tau}).
				\end{equation*}
				With probability $\bar{\epsilon}$, the history explores with the update rule: 
	\begin{equation}\begin{aligned}\label{historyUR4}
				&h_{\tau+1}=(h^R_{\tau},s_{\tau+1}),\ \ p(h_{\tau+1})=(p^R(h_{\tau}),J(s_{\tau+1})).
				\end{aligned}
				\end{equation}
			\end{enumerate}
			\\
			\hline
		\end{tabular}\label{PSBRD}
	\end{center}
\end{table}

\subsection{Background on Finite Markov Chains}\label{BackgroundFMC}
Before presenting our main results, we provide some preliminaries on finite Markov chains \cite{GCC:19,HPY:93}. In order to maintain consistency with our analysis later, we use the same set of symbols to represent the state space and the states as before in this subsection. Let $P^0$ be the transition matrix of a stationary Markov chain defined on a finite state space $\mathcal{H}$ and $P^{\epsilon}$ be a \emph{regular perturbation} of $P^0$ defined below \cite{HPY:93}.

\begin{defi}[Regular perturbation]\label{RegularPerDefi}
    A family of Markov processes $P^{\epsilon}$ is a regular perturbation of $P^0$ defined over $\mathcal{H}$, if there exists $\epsilon_1>0$ such that the following conditions hold for all $h,h'\in\mathcal{H}$:
    \begin{enumerate}[label=(\roman*)]
	\item\label{item:RegularP1} $P^{\epsilon}$ is aperiodic and irreducible for all $\epsilon\in(0,\epsilon_1]$; 
	\item\label{item:RegularP2} $\lim_{\epsilon\to0}P_{hh'}^{\epsilon}=P_{hh'}^0$;
	\item\label{item:RegularP3} $P_{hh'}^{\epsilon}>0$ for some $\epsilon>0$ implies that there exists an $r(h,h')\ge0$, called the \emph{resistance of the transition} from $h$ to $h'$, such that
	\begin{equation}\label{ResistanceDefi}
	0<\lim_{\epsilon\to0^+}\frac{P_{hh'}^{\epsilon}}{\epsilon^{r(h,h')}}<\infty.
	\end{equation} 
\end{enumerate}
\end{defi}

Note that the resistance $r(h,h')=0$ if and only if $P_{hh'}^0>0$. Take $r(h,h')=\infty$ if $P^\epsilon_{hh'}=0$. Let $\mathcal{G}_{\textup{H}}=(\mathcal{H},\mathcal{E}_{\textup{H}})$ be the digraph induced by the transition matrix $P^{\epsilon}$. For every pair of states $h,h'\in\mathcal{H}$, the directed edge $e_{hh'}$ exists in $\mathcal{E}_{\textup{H}}$ if $P_{hh'}^\epsilon>0$. Since $P^{\epsilon}$ is irreducible for $\epsilon>0$, it has a unique stationary distribution $\pi^{\epsilon}$ satisfying $\pi^{\epsilon\top} P^{\epsilon}=\pi^{\epsilon\top}$, $\pi^{\epsilon}_h>0$ and $\pi^{\epsilon\top}\mathbb{1}=1$, and the digraph $\mathcal{G}_{\textup{H}}$ is strongly connected. We consider the following concept of stability introduced in \cite{DF-HPY:90}.

\begin{defi}[Stochastic ability]\label{StableAbilityDefi}
	A state $h\in \mathcal{H}$ is stochastically stable relative to the process $P^{\epsilon}$ if $\lim_{\epsilon\to0}\pi_{h}^{\epsilon}>0$.
\end{defi}

Note that $\pi_h^{\epsilon}$ is the relative frequency with which state $h$ will be observed when the process runs for a very long time. Thus, over the long run, states that are not stochastically stable will be observed less frequently compared to states that are, provided that the perturbation $\epsilon$ is sufficiently small. Moreover, we emphasize here that  $\lim_{\epsilon\to0}\pi_{h}^{\epsilon}$ exists for all $h\in\mathcal{H}$, which will be shown in \lemaref{LimitDis}. 
It turns out that the $\mathcal{W}$-graph defined below is useful in analyzing stochastic stability of states.

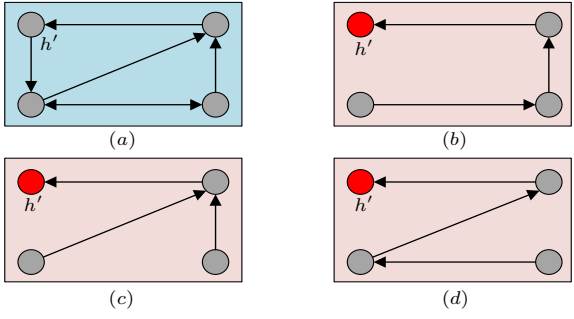
\begin{figure}
	\centering
	%%21.3,13.4
	\subfigure{
		\includegraphics[width=76mm,height=38mm]{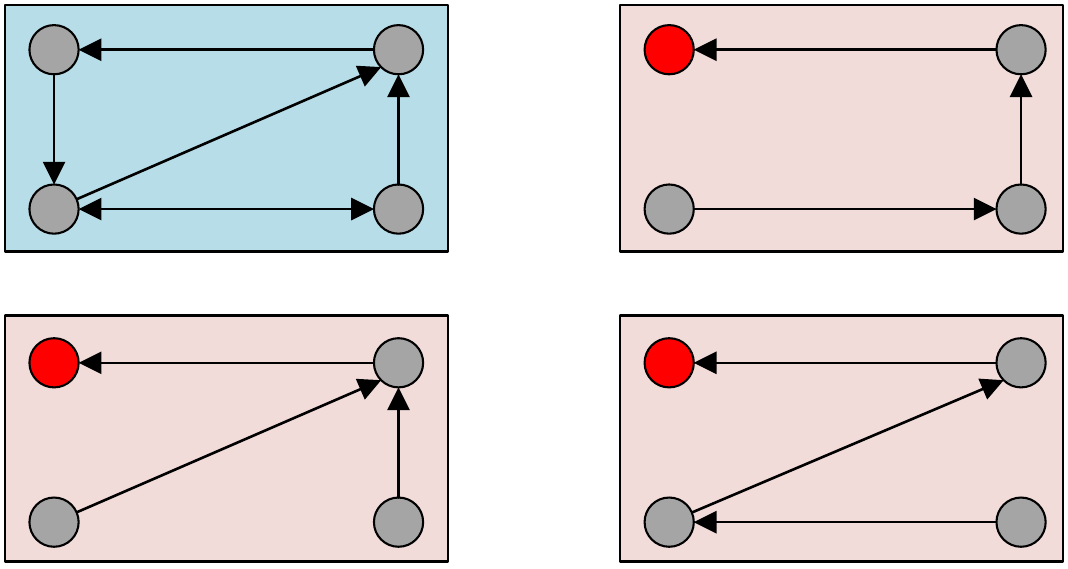}
		\put(-83,28){\scriptsize $h'$}
		\put(-202,89){\scriptsize $h'$}
		\put(-83,87){\scriptsize $h'$}
		\put(-208,28){\scriptsize $h'$}
		\put(-176,53){\scriptsize $(a)$}
		\put(-50,53){\scriptsize $(b)$}
		\put(-176,-7){\scriptsize $(c)$}
		\put(-50,-7){\scriptsize $(d)$}
	}
	\caption{An example of $h'$-trees in case $\mathcal{H}$ containing four states, where the graph $\mathcal{G}_{\textup{H}}$ in $(a)$ consists of three $h'$-trees for some $h'\in\mathcal{H}$ (in red) shown in $(b)$, $(c)$ and $(d)$.}
	\label{H_tree_figure}
\end{figure}

\begin{defi}[$\mathcal{W}$-graph] For any nonempty proper subset $\mathcal{W}\subset\mathcal{H}$, a digraph $g=(\mathcal{H},\mathcal{E}_g)$ is a $\mathcal{W}$-graph of $\mathcal{G}_{\textup{H}}=(\mathcal{H},\mathcal{E}_{\textup{H}})$ if it satisfies the following conditions:
\begin{enumerate}[label=(\roman*)]
        \item the edge set satisfies $\mathcal{E}_g\subset\mathcal{E}_{\textup{H}}\setminus\{e_{hh'}\,|\,h\in\mathcal{W},h'\in\mathcal{H},h\neq h'\}$;
		\item every node $h\in \mathcal{H}\setminus\mathcal{W}$ is the start node of exactly one edge in $\mathcal{E}_g$;
		\item there are no cycles in the graph $g$; or equivalently, for any node $h\in \mathcal{H}\setminus\mathcal{W}$, there exists a sequence of directed edges from $h$ to some state $h'\in\mathcal{W}$.
\end{enumerate}
\end{defi}

Note that if $\mathcal{W}$ is a singleton consisting of $h'$, then a $\mathcal{W}$-graph $g$ is actually a spanning tree of $\mathcal{G}_{\textup{H}}$ such that from every node $h\neq h'$ there is a unique path from $h$ to $h'$. We call such a $\mathcal{W}$-graph an $h'$-tree, as in \cite{HPY:93}. We will denote by $\mathbb{G}(\mathcal{W})$ the set of $\mathcal{W}$-graphs. An example of $h'$-trees for some $h'\in\mathcal{H}$ when $\mathcal{H}$ contains four states, is shown in Fig.~\ref{H_tree_figure}.

We now state Young's results \cite{HPY:93} for perturbed Markov
processes. Define the \emph{stochastic potential} of state $h$ by
\begin{equation}\label{StochasticPotentialeq}
\gamma(h)=\min_{g\in\mathbb{G}(h)}\sum_{e_{h'h''}\in \mathcal{E}_g}r(h',h''),
\end{equation} 
where $r(h',h'')$ is the resistance from $h'$ to $h''$ defined in \defiref{RegularPerDefi}.

\begin{lema}[Perturbation and stochastic stability {\cite[Lemma 1]{HPY:93}}]\label{LimitDis}
  Let $P^\epsilon$ be a regular perturbation of $P^0$ and $\pi^\epsilon$ be
  its stationary distribution. Then
  \begin{enumerate}
  \item $\lim_{\epsilon\to0}\pi^\epsilon=\pi^0$ where $\pi^0$ is a
    stationary distribution of $P^0$; and
  \item $h$ is stochastically stable ($\pi_h^0>0$) if and only if
    $\gamma(h)\leq\gamma(h')$ for all $h'$ in $\mathcal{H}$.
  \end{enumerate}
\end{lema}

\subsection{Policy Seeking}
In this subsection, we propose the class of perturbed SBRD algorithms as shown in \tabref{PSBRD} to seek the best policies according to the underlying metrics.

The unperturbed $(\epsilon=0)$ and perturbed $(\epsilon>0)$ SBRDs are two Markov chains over the set of history states $\mathcal{H}$. For consistency, we use the same symbols and concepts introduced in Section \ref{BackgroundFMC}. Let $P^0$ and $P^{\epsilon}$ be the transition matrices of the unperturbed and perturbed SBRDs, respectively. 

We first consider the unperturbed SBRD $P^0$ and discuss its connections
with the sink equilibria of the mega-game $G_2$. For simplicity, we use
$s\in h$ to indicate that a joint strategy $s$ is recorded in a history
state $h$. We let $v$ denote the number of sink equilibria in the meta-game
$G_2$ and $\mathcal{Q}=\{Q_1,Q_2,\dots,Q_{v}\}$ denote the set of sink
equilibria. For each sink equilibrium $Q_i\in\mathcal{Q}$, we define an
induced set $Y_i$, called \emph{recurrent communication class (RCC)}, of
the history states as follows:
\begin{equation}\label{Class2}
Y_i=\{h\in\mathcal{H}\,|\,h\text{ is an SBRP, }\forall s\in h, \text{ then }s\in Q_i\}.
\end{equation}
We summarize these induced RCCs as $\mathcal{Y}=\{Y_1,Y_2,\dots,Y_v\}$. Next, we discuss the dynamic stability of the unperturbed SBRD $P^0$.

\begin{pro}[Unperturbed process]\label{UnperturbedProcess}
	For the unperturbed SBRD $P^0$, each RCC $Y_i$ is an absorbing history state set, i.e., once the history state enters $Y_i$, it will not leave $Y_i$. As a result, $P^0$ is an absorbing chain with $v$ absorbing history state sets.  
\end{pro}
\begin{proof}
  When $\epsilon=0$, only steps \ref{item:SS}, \ref{item:RL}, and
  \ref{item:HU} of \tabref{PSBRD} are executed. In step \ref{item:HU}, if
  the latest joint strategy (i.e., $s_{\tau}$) in $h_{\tau}$ is a PNE, then
  $s_{\tau}$ will be appended. Thus after running for a finite number of
  times, $h_{\tau}$ only consists of this PNE, implying that $h_{\tau}\in
  Y_i$ for some $Y_i\in\mathcal{Y}$. If $s_{\tau}$ is not a PNE, then we
  record $s_{\tau+1}$ only when it is an SBRP from $s_{\tau}$. Thus, it is
  possible that after a finite run, $h_{\tau}$ will be filled with an SBRP
  with all involved joint strategies belonging to a non-singleton sink
  equilibrium, if no PNE is visited in this run. Thus, in this case,
  $h_{\tau}\in Y_i$ for some $Y_i\in\mathcal{Y}$ and it will never leave
  $Y_i$, where $Y_i$ is defined in \eqref{Class2}. From the above, each RCC
  $Y_i$ is an absorbing history state set. This concludes the proof.
\end{proof}

\proref{UnperturbedProcess} shows that when $P^0$ runs for a long period of time, the joint strategies in the memory subsystem all come from one sink equilibrium. However, it is hard to predict which sink equilibrium will be reached based on the initial joint strategy $s_1$. Thus, similar to the method in \cite{GCC:19}, we want to design a class of perturbations to $P^0$ such that after a long run, the perturbed SBRD $P^{\epsilon}$ guarantees that the joint strategies with the maximum underlying metrics are observed in the memory subsystem with high probability, regardless of the initial joint strategies. In this paper, we use stochastic stability as a solution concept, as in \cite{GCC:19,JRM:17,HB-JM-JS:19,ACC-DSL-AR-NRJ:13}. 

 First, motivated by the concept of mistake in \cite{HPY:93}, we introduce the concept of \emph{exploration number} for the transition between two joint strategies, which plays a key role in the following analysis.

\begin{defi}[Exploration number]\label{EXdefi}
  For any two joint strategies $s$ and $s'$, the exploration number
  $e(s,s')\in\mathbb{N}$ from $s$ to $s'$ is the minimum number of agents
  required to explore in order to achieve the transition under the SBRD
  $P^{\epsilon}$ from $s$ to $s'$, plus one if the history update has to
  explore to attach $s'$ through \eqref{historyUR4}.
\end{defi}

For example, consider a 3-agent case and take two joint strategies $s=(s^1,s^2,s^3)$ and $s'=(s'^1,s'^2,s'^3)$. Suppose that $s$ is a PNE, $s^1=s'^1,s^2\neq s'^2,s^3\neq s'^3$, and $s'^2,s'^3$ are not best responses with respect to $s$, i.e., $s'^2\notin B^2(s^{-2})$ and $s'^3\notin B^3(s^{-3})$. Note that both agent $2$ and agent $3$ have to explore to achieve the transition from $s$ to $s'$. Additionally, the history update has to explore to attach $s'$ through \eqref{historyUR4}, as $s$ is a PNE. Thus, the exploration number from $s$ to $s'$ is given by $e(s,s')=2+1=3$.    

Next, we show that $P^{\epsilon}$ is a regular perturbation of $P^0$, and give a bound for the resistance between two history states.

\begin{lema}[Perturbed process]\label{RegularPer}
	The perturbed SBRD $P^{\epsilon}$ is a regular perturbation of $P^0$ over the set of history states $\mathcal{H}$, and  the resistance $r(h,h')$ of moving from $h$ to $h'$ satisfies:
	\begin{enumerate}[label=(\roman*)]
		\item\label{item:PerturbedP1} if $P^{\epsilon}_{hh'}>0$ and $P^0_{hh'}>0$, then $r(h,h')=0$;
		\item\label{item:PerturbedP2} if $P^{\epsilon}_{hh'}>0$ and $P^0_{hh'}=0$, then $f(p(h))\leq r(h,h')\leq (n+1)f(p(h))$;
		\item\label{item:PerturbedP3} if $P^{\epsilon}_{hh'}=0$, then $r(h,h')=\infty$.
	\end{enumerate}
\end{lema}

\begin{proof}	
For the perturbed SBRD $P^{\epsilon}$ in \tabref{PSBRD}, the strategy selection rule implies a positive transition probability between any two joint strategies and the history update indicates that the probability of adjoining a new joint strategy is also positive. Thus, it is possible to get to any history state from any history state in a finite number of transitions, which implies that $P^\epsilon$ is irreducible. The history update also guarantees that there always exists a history state $h\in\mathcal{H}$ such that $P^{\epsilon}_{hh}>0$, no matter whether a PNE exists in the meta-game $G_2$. Thus, $P^{\epsilon}$ is aperiodic. Moreover, $\lim_{\epsilon\to 0}P^{\epsilon}_{hh'}=P^0_{hh'}$ is straightforward. We check condition \ref{item:RegularP3} in \defiref{RegularPerDefi} below.

Regarding~\ref{item:PerturbedP1}: 
If $P_{hh'}^{\epsilon}>0$, we have $h^R=h'^L$. Furthermore, if $P^0_{hh'}>0$, then $r(h,h')=0$, i.e., there is no resistance, such that \eqref{ResistanceDefi} holds. Thus, \ref{item:PerturbedP1} is obtained.

Regarding~\ref{item:PerturbedP2}, we denote $h=(h^L,s)=(\bar{s},h^R)$ and $h'=(h'^L,s')$. Then, $P_{hh'}^{\epsilon}>0$ means $h^R=h'^L$, and $P^0_{hh'}=0$ means that from $s$ to $s'$ there is at least one agent who explores, that is, $e(s,s')\ge1$. Notice that \defiref{EXdefi} guarantees that $e(s,s')\leq n+1$. By following the computation of $P_{hh'}^{\epsilon}$ in \cite[Section 5]{HPY:93} or \cite[Section 5.2]{ACC-DSL-AR-NRJ:13}, the resistance from $h$ to $h'$ is given by $r(h,h')=e(s,s')f(p(h))$, which guarantees that \eqref{ResistanceDefi} holds. Thus, we have $f(p(h))\leq r(h,h')\leq (n+1)f(p(h))$.

Regarding \ref{item:PerturbedP3}, since $P^{\epsilon}_{hh'}=0$, then
$r(h,h')=\infty$ such that \eqref{ResistanceDefi} holds. From the above, we
conclude that $P^{\epsilon}$ is a regular perturbation of $P^0$. This
completes the proof.
\end{proof}

Since $P^{\epsilon}$ is irreducible, i.e., $\mathcal{G}_{\textup{H}}$ is strongly connected, we can define the resistance of a path in $\mathcal{G}_{\textup{H}}$ from one node to another as the sum of the resistances of the edges along the path. According to \proref{UnperturbedProcess} and \lemaref{RegularPer}, the RCCs in $\mathcal{Y}$ are disjoint, and under the perturbed process $P^{\epsilon}$ they satisfy the following three properties as in \cite{HPY:93}: $(i)$ from every history state there is a path of zero resistance to at least one of the RCCs $Y_i$; $(ii)$ within each RCC $Y_i$ there is a path of zero resistance from every history state to every other; $(iii)$ every edge exiting from $Y_i$ has positive resistance.

We define the performance of each history state $h\in \mathcal{H}$ as the average performance of the joint strategies in $h$:
\begin{equation}\label{HistoryWelfare}
W(h)=\frac{1}{m}\sum_{s\in h}W(s),
\end{equation}
where $W(s)$ is defined in \eqref{JointMeasure} and characterizes the overall performance of a strategy. Furthermore, the performance of each RCC $Y_i$ is defined as the worst performance of the history states in $Y_i$:
\begin{equation}\label{RecurrentMeasure}
W(Y_i)=\min_{h\in Y_i}W(h).
\end{equation}

Next we establish the connections between metrics of the strategies proposed in Section \ref{MEvaluationSec} and the performances of the RCCs. Let $L$ be the maximum number of joint strategies in a directed cycle $\mathcal{L}$ for all sink equilibria in $\mathcal{Q}$, i.e.,
\begin{equation}\label{MaxLengthCycle}
L=\max_{Q\in\mathcal{Q}}\max_{\mathcal{L}\in\mathcal{C}(Q)}|\mathcal{L}|.
\end{equation}

\begin{figure}
	\centering
	%%21.3,13.4
	\subfigure{
		\includegraphics[width=78mm,height=24mm]{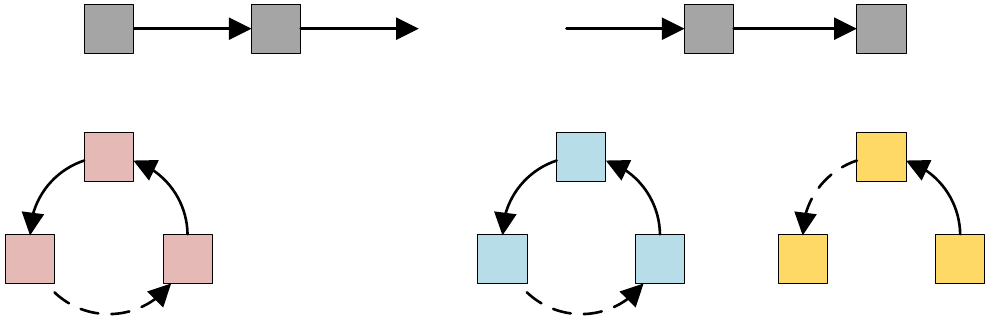}
		\put(-26,16){\scriptsize $\bar{h}$}
		\put(-96,16){\scriptsize $\mathcal{L}_K$}
		\put(-201,16){\scriptsize $\mathcal{L}_1$}
		\put(-110,42){\rotatebox{90}{\large $\Leftarrow$}}
		\put(-206,-6){\scriptsize $z_1(h)$}
		\put(-102,-6){\scriptsize $z_K(h)$}
		\put(-153,16){\Large $\cdots$}
		\put(-120,58){\Large $\cdots$}
		\put(-201,51){\scriptsize $s_1$}
		\put(-163,51){\scriptsize $s_2$}
		\put(-70,51){\scriptsize $s_{m-1}$}
		\put(-29,51){\scriptsize $s_m$}
		\put(-218,60){\scriptsize $h$}
	}
	\caption{Each history state $h$ in an RCC $Y_i$ can be decomposed into $z_1(h)\in\mathbb{N}$ directed cycles of type $\mathcal{L}_1$, $\dots$, $z_K(h)\in\mathbb{N}$ directed cycles of type $\mathcal{L}_K$ and a remainder $\bar{h}$, where $\bar{h}$ is an SBRP without containing any cycle (if nonempty) and $\mathcal{C}(Q_i)=\{\mathcal{L}_1,\dots,\mathcal{L}_K\}$ is the set of directed cycles in $Q_i$.}
	\label{Connections_figure}
\end{figure}
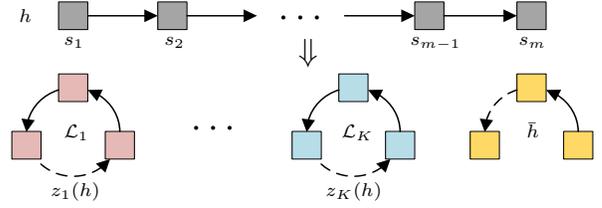

\begin{lema}[Connections between performance and metrics]\label{ReCClass} For each RCC $Y_i\in\mathcal{Y}$, the following conditions hold:
	\begin{enumerate}[label=(\roman*)]
		\item\label{item:CnectPM1} if the cycle-based metric is considered, then for any $\delta>0$, there exists
		\begin{equation}\label{mboundequ}
		\bar{m}_{\delta}=\frac{LJ_{\rm max}}{\delta},
		\end{equation}
		such that if $m\ge \bar{m}_{\delta}$, then
		\begin{equation}\label{PerformaceCycle1}
		|W(Y_i)-M_c(Q_i)|=|W(Y_i)-M_c(s)|\leq \delta,
		\end{equation}
		for all joint strategies $s\in Q_i$;
		\item\label{item:CnectPM2} if the memory-based metric is considered, then $W(Y_i)=M_m(Q_i)=M_m(s)$ for all joint strategies $s\in Q_i$.
	\end{enumerate}
\end{lema}
\begin{proof}

Regarding \ref{item:CnectPM1}, if $Q_i$ is a singleton containing a unique PNE $s$, then \eqref{Class2} indicates that $Y_i$ is also a singleton given by $\{(s,s,\dots,s)\}$. Thus, according to \eqref{HistoryWelfare} and \eqref{RecurrentMeasure}, $W(Y_i)=W(s)$. On the other hand, it follows from \eqref{PathMeasure} and \defiref{CBMeasure} that $M_c(s)=W(s)$, as $\mathcal{C}(Q_i)=\{s\}$. Thus, $W(Y_i)=M_c(s)$, which satisfies \eqref{PerformaceCycle1}.

If $Q_i$ is not a singleton, denote the set of directed cycles in $Q_i$ by $\mathcal{C}(Q_i)=\{\mathcal{L}_1,\mathcal{L}_2,\dots,\mathcal{L}_K\}$ $(K\in\mathbb{N}_{>0})$, where $\mathcal{L}_k$ $(k\in[K])$ is a directed cycle. Next, we show that the performance of each history state $h$ in $Y_i$ can be computed as follows:
\begin{equation}\label{HisWExpression}
W(h)=\frac{\sum_{k=1}^Kz_k(h)|\mathcal{L}_k|W(\mathcal{L}_k)+\sum_{s\in \bar{h}}W(s)}{m},
\end{equation}
where $z_k(h)\in\mathbb{N}$ and $\sum_{k=1}^Kz_k(h)|\mathcal{L}_k|+|\bar{h}|=m$. The path $\bar{h}$ is the remainder of $h$ when all directed circles contained in $h$ are removed. The expansion \eqref{HisWExpression} implies that the history state $h$ in $Y_i$ can be decomposed into many directed cycles (if exist) and a remainder without containing any cycle (if nonempty), as Fig. \ref{Connections_figure} shows. The details are as follows.

First, it follows from \eqref{Class2} that each history state $h\in Y_i$ can be expressed as an SBRP (we take the starting time as $\tau=1$):
\begin{equation*}
h=s_1\to s_2\to\cdots\to s_m.
\end{equation*}
We move the joint strategy starting from $s_1$ along the path $h$. If the same joint strategy is revisited at time $t_2$, i.e., there exists a time instant $t_1\in[1,t_2)$ such that $s_{t_1}=s_{t_2}$. Then the path $s_{t_1}\to\cdots\to s_{t_2}$ is a directed cycle of $Q_i$. We remove this directed cycle and obtain the remainder
\begin{equation}
\bar{h}_1=s_1\to \cdots \to s_{t_1}\to s_{t_2+1}\to\cdots\to s_m.
\end{equation}
Then, we continue to move the joint strategy starting from $s_{t_1}$ along the path $\bar{h}_1$. Similarly, we remove the directed cycle and continue to move along the remainder $\bar{h}_2$. When $s_m$ is reached, this method decomposes $h$ into many directed cycles (if exist) and a remainder $\bar{h}$ without containing any cycle (if nonempty). Thus, \eqref{HisWExpression} is verified. Moreover, the remainder $\bar{h}$ is a part of a directed circle if nonempty, implying that $|\bar{h}|\leq L$. 

Denote by  $\mathcal{L}^*$ the directed cycle in $\mathcal{C}(Q_i)$ with the worst performance, i.e., $\mathcal{L}^*\in\arg\min_{\mathcal{L}\in\mathcal{C}(Q_i)}W(\mathcal{L})$, where $W(\mathcal{L})$ is computed by \eqref{PathMeasure}. Thus, \eqref{CycleMeasureEqu} leads to $M_c(s)=W(\mathcal{L}^*)$ for all joint strategies $s\in Q_i$. By \eqref{JointMeasure}, \eqref{PathMeasure} and \aspref{PayoffBound}, we obtain 
\begin{equation}\label{BoundJointCycle}
0\leq W(s)\leq J_{\rm max},\quad 0\leq W(\mathcal{L}^*)\leq J_{\rm max}.
\end{equation}

For any $\delta>0$, let $\bar{m}_\delta$ be as in \eqref{mboundequ}, and then it follows from \eqref{RecurrentMeasure}, \eqref{MaxLengthCycle}, \eqref{HisWExpression} and \eqref{BoundJointCycle} that for any memory $m\ge \bar{m}_\delta>0$ and $s\in Q_i$, we have
\begin{equation}\begin{aligned}\label{RCCdelta1}
W(Y_i)&=\min_{h\in Y_i}W(h)\\
&=\min_{h\in Y_i}\frac{\sum_{k=1}^Kz_k(h)|\mathcal{L}_k|W(\mathcal{L}_k)+\sum_{s\in \bar{h}}W(s)}{m}\\
&=\frac{\sum_{k=1}^Kz_k(h')|\mathcal{L}_k|W(\mathcal{L}_k)+\sum_{s\in \bar{h}'}W(s)}{m}\\
& \ge \frac{\sum_{k=1}^Kz_k(h')|\mathcal{L}_k|W(\mathcal{L}^*)}{m}\\
& =W(\mathcal{L}^*)-\frac{|\bar{h}'|W(\mathcal{L}^*)}{m}\ge W(\mathcal{L}^*)-\frac{|\bar{h}'|J_{\rm max}}{m}\\
& \ge W(\mathcal{L}^*)-\frac{LJ_{\rm max}}{m} \ge M_c(s)-\delta,
\end{aligned}
\end{equation}
where $h'\in\arg\min_{h\in Y_i}W(h)$, $\sum_{k=1}^Kz_k(h')|\mathcal{L}_k|+|\bar{h}'|=m$, $|\bar{h}'|\leq L$ and $M_c(s)=W(\mathcal{L}^*)$ for $s\in Q_i$ are used.

On the other hand, we consider a history state $h^*$ in $Y_i$ which can be decomposed into $z(h^*)\in\mathbb{N}$ directed circles $\mathcal{L}^*$ (if exist) and a remainder $\bar{h}^*$ with $|\bar{h}^*|\leq L$. Then, according to \eqref{HistoryWelfare}, \eqref{RecurrentMeasure}, \eqref{HisWExpression} and \eqref{BoundJointCycle}, for any memory $m\ge \bar{m}_\delta>0$ and $s\in Q_i$, we have
\begin{equation}\begin{aligned}\label{RCCdelta2}
W(Y_i)&=\min_{h\in Y_i}W(h)\leq W(h^*)\\
&=\frac{z(h^*)|\mathcal{L}^*|W(\mathcal{L}^*)+\sum_{s\in \bar{h}^*}W(s)}{m}\\
&\leq\frac{z(h^*)|\mathcal{L}^*|W(\mathcal{L}^*)+|\bar{h}^*|J_{\rm max}}{m}\\
&= W(\mathcal{L}^*)+\frac{|\bar{h}^*|(J_{\rm max}-W(\mathcal{L}^*))}{m}\\
&\leq W(\mathcal{L}^*)+\frac{LJ_{\rm max}}{m}\leq  M_c(s)+\delta,
\end{aligned}
\end{equation}
where $z(h^*)|\mathcal{L}^*|+|\bar{h}^*|=m$ and $|\bar{h}^*|\leq L$ are used. Therefore, \eqref{PerformaceCycle1} follows from \eqref{RCCdelta1} and \eqref{RCCdelta2}.

Regarding \ref{item:CnectPM2}, according to the definition, $\mathcal{M}(Q_i)$ is the set of SBRPs of length $m$ in the sink equilibrium $Q_i$. According to \eqref{Class2}, we obtain that $Y_i=\mathcal{M}(Q_i)$. Thus, it follows from \eqref{PathMeasure}, \eqref{MemoryMeasureEqu}, \eqref{HistoryWelfare} and \eqref{RecurrentMeasure} that $W(Y_i)=M_m(s)$ for $s\in Q_i$, which finishes the proof.
\end{proof}

Note that \lemaref{LimitDis} reveals that the stochastic stability of a
history state depends crucially on its stochastic potential defined in
\eqref{StochasticPotentialeq}. Next, we give lower and upper bounds on the
stochastic potentials of history states in the RCCs; these bounds play an
important role in our following analysis.

\begin{lema}[Bounds on the stochastic potential]\label{StocPotential}
  For each RCC $Y_i\in\mathcal{Y}$, the stochastic potential of a history
  state $h$ in $Y_i$ under the perturbed SBRD $P^{\epsilon}$, satisfies the
  following properties:
  \begin{enumerate}[label=(\roman*)]
  \item\label{item:StocP1} $\gamma(h)=\gamma(h')$ for all history states
    $h'\in Y_i$;
  \item\label{item:StocP2} $\gamma(h)$ satisfies $\bar{\gamma}_i\leq
    \gamma(h)\leq (n+1)\bar{\gamma}_i$, where
    \begin{equation*}
      \bar{\gamma}_i=\sum_{Y_j\in\mathcal{Y}\setminus\{Y_i\}}\min_{h'\in Y_j}f(p(h')).
    \end{equation*}
  \end{enumerate}
\end{lema}
\begin{proof}
Regarding \ref{item:StocP1}, it directly follows from Lemma 2 in \cite{HPY:93}, that is, all history states in an RCC $Y_i$ have the same stochastic potential. 

Regarding \ref{item:StocP2}, in any $h$-tree $g$, there is at least one history state $h'_j$ in other RCC $Y_j$ $(j\neq i)$ which needs to explore. It follows from \ref{item:PerturbedP2} in \lemaref{RegularPer} that the resistance of this exploration is not less than $f(p(h'_j))$. Thus, by considering \eqref{StochasticPotentialeq} and all RCCs $Y_j\in\mathcal{Y}$ except $Y_i$, we have
\begin{equation}\label{SPequ1}
\gamma(h)\ge \sum_{j\in[v]\setminus\{i\}}f(p(h'_j))\ge \sum_{Y_j\in\mathcal{Y}\setminus\{Y_i\}}\min_{h'\in Y_j}f(p(h')).
\end{equation}

Without loss of generality, we consider the case $i=1$, i.e., the RCC $Y_1$. Next, we will construct a special $h$-tree $g_0$ rooted at some node $h\in Y_1$. In the unperturbed SBRD $P^0$, all nodes (history states) in $\mathcal{H}$ can be classified into two categories: Not-in-a-sink $\mathcal{H}_{\rm not}$ (blue) and In-a-sink $\mathcal{H}_{\rm in}$ (pink), as Fig. \ref{SP_figure}(a) shows. Once the history state enters $\mathcal{H}_{\rm in}$, it will never leave it if no exploration happens.

For the In-a-sink $\mathcal{H}_{\rm in}$, it can be decomposed into $v$ smaller disjoint parts, each of which corresponds to an RCC. In the example shown in Fig. \ref{SP_figure}(a), there are three RCCs: $Y_1$ (black), $Y_2$ (orange) and $Y_3$ (yellow), i.e., $v=3$.

Each RCC $Y_j$ has a directed spanning tree $g_j$ rooted at a node $h_j^*$ in $Y_j$, where $h_j^*\in\arg\min_{h'\in Y_j}f(p(h'))$. Also notice that all edges in $g_j$ have zero resistance. For this example, as Fig. \ref{SP_figure}(b) shows, the node $h_j^*$ is the white circle and the edges in $g_j$ are in red. Since  in this example $Y_1$ is a singleton (i.e., $h_1^*$), thus $g_1$ is empty.

Based on the unperturbed SBRD $P^0$, we classify the Not-in-a-sink $\mathcal{H}_{\rm not}$ into several smaller disjoint parts as follows: the first part $\mathcal{H}_1$ is the set of history states in $\mathcal{H}_{\rm not}$ from which there is a path of zero resistance to $Y_1$; the second part $\mathcal{H}_2$ is the set of history states in $\mathcal{H}_{\rm not}\setminus \mathcal{H}_1$ from which there is a path of zero resistance to $Y_2$; the third part $\mathcal{H}_3$ is the set of history states in $\mathcal{H}_{\rm not}\setminus(\mathcal{H}_1\cup \mathcal{H}_2)$ from which there is a path of zero resistance to $Y_3$; in this way, $\mathcal{H}_{\rm not}$ is divided into $v$ disjoint parts $\mathcal{H}_1,\mathcal{H}_2,\dots,\mathcal{H}_v$, and $\mathcal{H}_j$ may be empty. For the part $\mathcal{H}_j$ (if nonempty), we can choose the edges associated with $\mathcal{H}_j$ such that for all $h\in \mathcal{H}_j$, there is a unique path of zero resistance without cycle from $h$ to $Y_j$. For this example, as Fig. \ref{SP_figure}(c) shows, $\mathcal{H}_{\rm not}$ is divided into three parts: $\mathcal{H}_1$ (red), $\mathcal{H}_2$ (green) and $\mathcal{H}_3$ (purple). The transition paths of zero resistance from $h\in\mathcal{H}_j$ to $Y_j$ are in red.

Finally, we build the paths from the RCCs $Y_j\in\mathcal{Y}\setminus Y_1$ to $Y_1$. Note that for any $h\in \mathcal{H}$, if the rightmost $s$ of $h$ belongs to $Q_1$, then $h\in \mathcal{H}_1\cup Y_1$, because finite transitions of zero resistance can lead $h$ to $Y_1$. Since the transition probability between any two states and the probability of attaching a new one are both positive, so $h_j^*$ can adjoin any joint strategy $s$. Thus, $h_j^*$ can reach one node of $\mathcal{H}_1\cup Y_1$ by one experiment to attach a joint strategy $s$ in $Q_1$. Then, $h_j^*$ can reach $Y_1$ by continuing to move along the path by which this intermediate node moves to $Y_1$. By this way, we are able to construct an $h_1^*$-tree $g_0$. For our example, the $h_1^*$-tree $g_0$ is shown in Fig. \ref{SP_figure}(d), where the one-step transitions of nonzero resistance from $h_j^*$ $(j\neq1)$ to one node of $\mathcal{H}_1\cup Y_1$ are in dashed red lines.

Note that in the constructed $h_1^*$-tree $g_0$, there are $v-1$ explorations each of which is not greater than $(n+1)f(h_j^*)$ by the conclusion \ref{item:PerturbedP2} in \lemaref{RegularPer}. Since $\gamma(h_1^*)$ is not greater than the sum of the resistances of $g_0$ by \eqref{StochasticPotentialeq}, we have
\begin{equation}\begin{aligned}\label{SPequ2}
\gamma(h_1^*)&\leq \sum_{j\neq 1}(n+1)f(p(h_j^*))\\
&=	(n+1)\sum_{Y_j\in\mathcal{Y}\setminus\{Y_1\}}\min_{h'\in Y_j}f(p(h')),
\end{aligned}
\end{equation} 
where $h_j^*\in\arg\min_{h'\in Y_j}f(p(h'))$ is used. Thus, the conclusion follows from \eqref{SPequ1}, \eqref{SPequ2} and the conclusion \ref{item:StocP1}.
\end{proof}

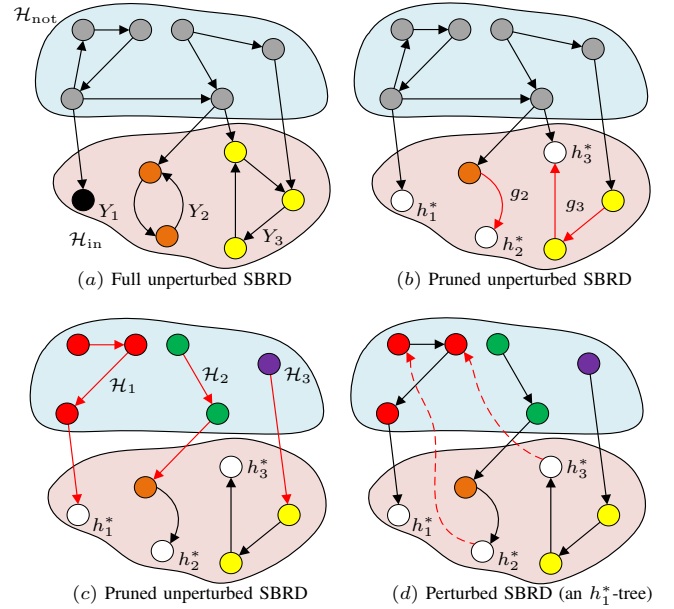
\begin{figure}
	\centering
	%%21.3,13.4
	\subfigure{
		\includegraphics[width=40mm,height=36mm]{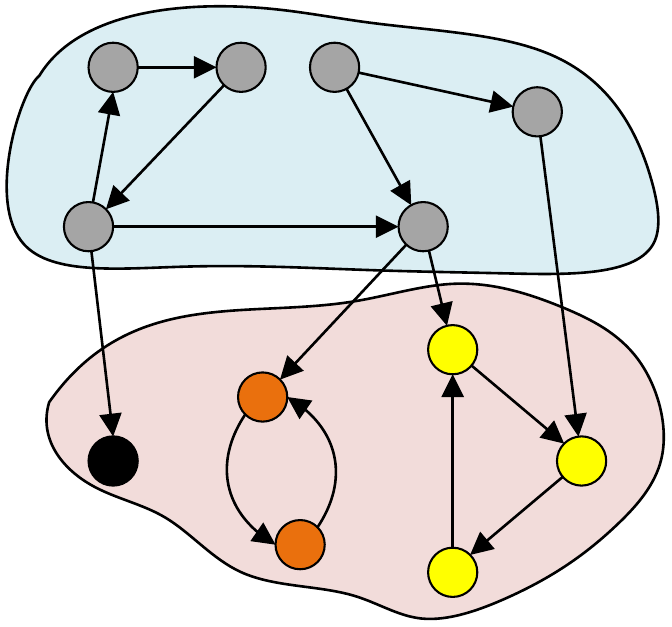}
		\put(-89,20.5){\scriptsize $Y_1$}
		\put(-55,20.5){\scriptsize $Y_2$}
		\put(-27,11){\scriptsize $Y_3$}
		\put(-121,95){\scriptsize $\mathcal{H}_{\rm not}$}
		\put(-100,10){\scriptsize $\mathcal{H}_{\rm in}$}
		\put(-97,-5){\scriptsize $(a)$ Full unperturbed SBRD}
	}
	\subfigure{
		\includegraphics[width=40mm,height=36mm]{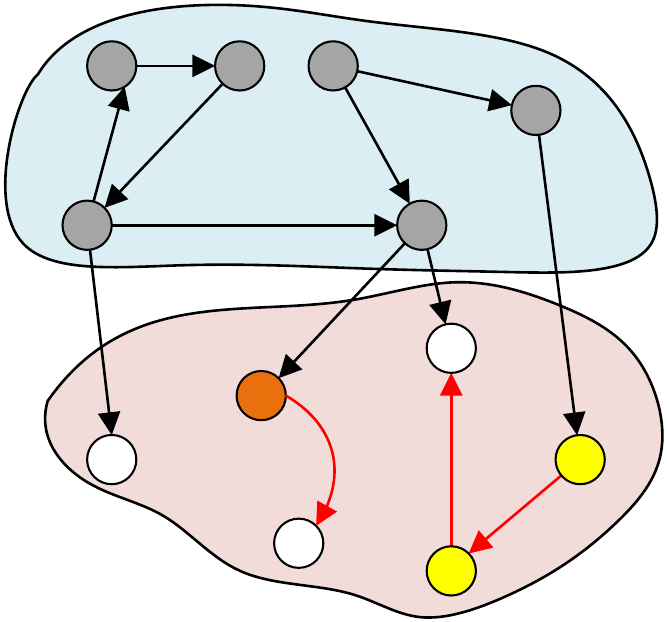}
		\put(-89.5,21){\scriptsize $h_1^*$}
		\put(-57,8){\scriptsize $h_2^*$}
		\put(-31,43){\scriptsize $h_3^*$}
		\put(-54,28){\scriptsize $g_2$}
		\put(-33,24){\scriptsize $g_3$}
		\put(-97,-5){\scriptsize $(b)$ Pruned unperturbed SBRD}
	}
	\subfigure{
		\includegraphics[width=40mm,height=36mm]{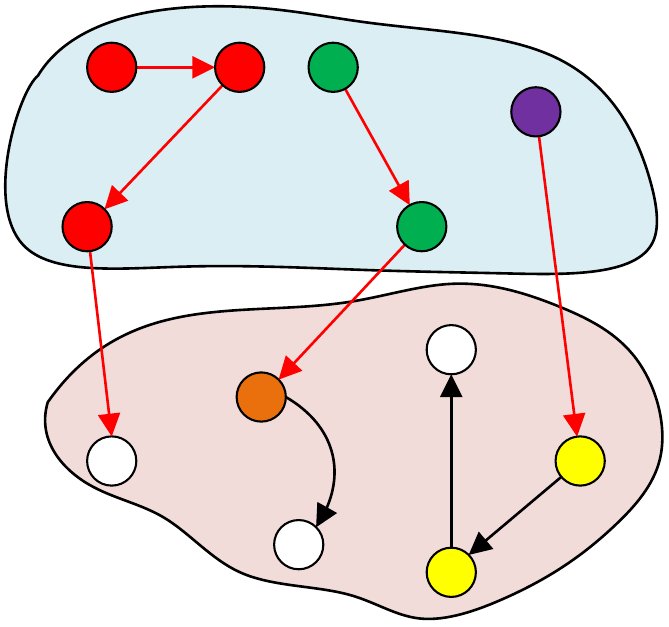}
		\put(-89.5,21){\scriptsize $h_1^*$}
		\put(-57,8){\scriptsize $h_2^*$}
		\put(-31,43){\scriptsize $h_3^*$}
		\put(-17,78){\scriptsize $\mathcal{H}_3$}
		\put(-48,78){\scriptsize $\mathcal{H}_2$}
		\put(-83,74){\scriptsize $\mathcal{H}_1$}
		\put(-97,-5){\scriptsize $(c)$ Pruned unperturbed SBRD}
	}
	\subfigure{
		\includegraphics[width=40mm,height=36mm]{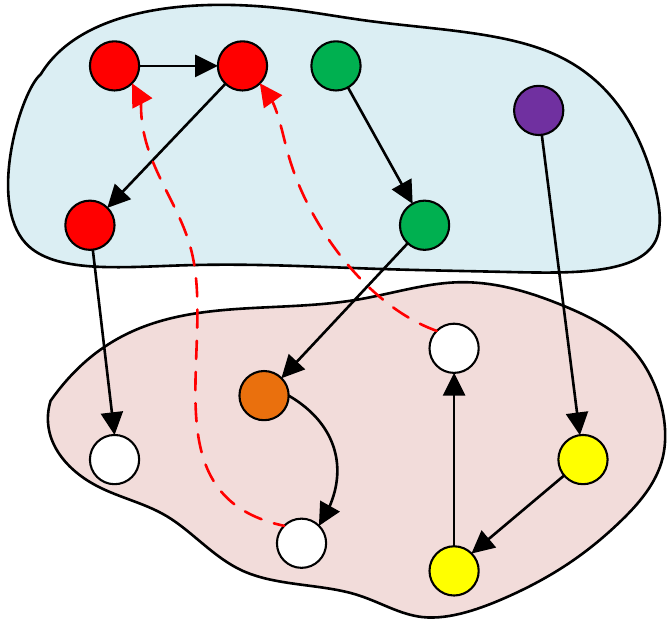}
		\put(-89.5,21){\scriptsize $h_1^*$}
		\put(-57,8){\scriptsize $h_2^*$}
		\put(-31,43){\scriptsize $h_3^*$}
		\put(-97,-5){\scriptsize $(d)$ Perturbed SBRD (an $h_1^*$-tree)}
	}
	\caption{An example for computing an upper bound of the stochastic potential for a history state in an RCC.}
	\label{SP_figure}
\end{figure}

%\fbtodo{did we define a $\delta$-function somewhere before?}

\begin{defi}[$\delta$-function]\label{deltafunc}
	For any $\delta>0$,  a function $f_{\delta}:\mathbb{R}\to\mathbb{R}_{>0}$ is a $\delta$-function, if for any $x\in \mathbb{R}$,
	\begin{enumerate}[label=(\roman*)]
		\item\label{item:DeltaFunc1} $f_\delta(x)$ is non-decreasing in $x$;
		\item\label{item:DeltaFunc2} $f_{\delta}(x+\delta)> (vn-n+1)f_{\delta}(x)$.
	\end{enumerate}
\end{defi}

\begin{defi}[Feasible function]\label{FeasibleFunc}
	For any $\delta>0$, a function $f:\mathbb{R}^{n\times
		m}\to\mathbb{R}_{>0}$ is \emph{$\delta$-feasible} if there exists a
	$\delta$-function $f_{\delta}$ such that for
	any $M=\{M_{ij}\}\in\mathbb{R}^{n\times m}$, the equality $f(M)=f_{\delta}(\frac{1}{m}\sum_{j=1}^m\sum_{i=1}^n w_iM_{ij})$ holds.
\end{defi}

We can design a $\delta$-feasible function $f$ by constructing a $\delta$-function $f_\delta$. For example, for any $M\in\mathbb{R}^{n\times m}$, we select
\begin{equation}\label{FeasibleFunExample}
f(M)=f_{\delta}(x)=(vn-n+2)^{\frac{x}{\delta}},
\end{equation}
where $x=\frac{1}{m}\sum_{j=1}^m\sum_{i=1}^n w_iM_{ij}$. It can be verified that $f$ in \eqref{FeasibleFunExample} is a $\delta$-feasible function.

We have the following assumption on the sink equilibrium. 

\begin{asp}[Uniqueness]\label{AspUmw}
	Suppose that the sink equilibrium $Q^*$ with the maximum metric $M_c(Q^*)$ or $M_m(Q^*)$ is unique. 
\end{asp}

Next, we present the main results on stochastic stability of the perturbed SBRD and propose a class of perturbations such that the sink equilibria with the maximum cycle-based metric can be observed frequently. 

\begin{thom}[Maximum cycle-based metric]\label{ThomMCBM}
	Let $G_2=(N,\mathcal{S},J)$ be an $n$-agent meta game on a finite joint strategy space $\mathcal{S}$. The training system has a memory of length $m\in\mathbb{N}_{>0}$ and consider the cycle-based metric $M_c(s)$ in \eqref{CycleMeasureEqu} for each joint strategy $s\in\mathcal{S}$. Suppose that Assumptions \ref{PayoffBound} and \ref{AspUmw} hold. If the following conditions are all true:
	\begin{enumerate}[label=(\Roman*)]
		\item\label{item:MCMcond1} there exists a $\delta_0>0$ such that $M_c(Q^*)\ge M_c(Q_i)+\delta_0$ for all sink equilibria $Q_i\in\mathcal{Q}\setminus\{Q^*\}$;
		\item\label{item:MCMcond2} $f$ is a $\delta$-feasible function for some $\delta\in(0,\delta_0)$;
		\item\label{item:MCMcond3} the memory $m$ is
greater than or equal to
		\begin{equation}\label{boundmMwequ}
		\bar{m}_{\delta}=\frac{2LJ_{\rm max}}{\delta_0-\delta},
		\end{equation}
	\end{enumerate}
	then under the perturbed SBRD $P^{\epsilon}$, 
	\begin{enumerate}[label=(\roman*)]
		\item\label{item:MCMconc1} a history state $h$ is stochastically stable if and only if $h\in Y^*$;
		\item\label{item:MCMconc2} as $\epsilon\to 0$, a joint strategy $s\in\mathcal{S}$ is recorded in the memory subsystem with nonzero probability if and only if $s\in Q^*$,
	\end{enumerate}
	where $M_c(Q^*)=\max_{Q\in \mathcal{Q}}M_c(Q)$ and $Y^*$ is the induced RCC of $Q^*$.
\end{thom}
\begin{proof}

Regarding \ref{item:MCMconc1}, for any $0<\delta<\delta_0$, let $\bar{m}_\delta$ be as in \eqref{boundmMwequ}. Then if $m\ge\bar{m}_\delta$, it follows from \ref{item:CnectPM1} in \lemaref{ReCClass} that
\begin{equation}\label{RRCQequ}
|W(Y_i)-M_c(Q_i)|=|W(Y_i)-M_c(s)|<\frac{\delta_0-\delta}{2},
\end{equation}
for all joint strategies $s\in Q_i$ and all RCCs $Y_i\in\mathcal{Y}$. 

For any $Y_i\in\mathcal{Y}\setminus\{Y^*\}$, it follows from \eqref{RRCQequ} and condition \ref{item:MCMcond1} that
\begin{equation}\begin{aligned}\label{RCCdisequ}
W(Y^*)&>M_c(Q^*)-\frac{\delta_0-\delta}{2}\ge M_c(Q_i)+\delta_0-\frac{\delta_0-\delta}{2}\\
&=M_c(Q_i)+\frac{\delta_0-\delta}{2}+\delta>W(Y_i)+\delta.
\end{aligned}
\end{equation}

Since $f$ is a $\delta$-feasible function, it follows from  \eqref{JointMeasure}, \eqref{HistoryWelfare} and \defiref{FeasibleFunc} that there exists a $\delta$-function $f_{\delta}$ such that 
\begin{equation}\label{change}
f(p(h))=f_{\delta}(W(h)),
\end{equation}
for all $h\in \mathcal{H}$. It follows from \eqref{RecurrentMeasure}, \eqref{RCCdisequ}, \eqref{change} and properties \ref{item:DeltaFunc1} and \ref{item:DeltaFunc2} in \defiref{deltafunc} that
\begin{equation}\begin{aligned}\label{Fdifferenceequ}
\min_{h\in Y_i}f(p(h))&=\min_{h\in Y_i}f_{\delta}(W(h))< \min_{h\in Y_i}\frac{f_\delta(W(h)+\delta)}{vn-n+1}\\
&=\frac{f_\delta(\min_{h\in Y_i}W(h)+\delta)}{vn-n+1}=\frac{f_\delta(W(Y_i)+\delta)}{vn-n+1}\\
&\leq\frac{f_\delta(W(Y^*))}{vn-n+1}= \frac{f_\delta(\min_{h\in Y^*}W(h))}{vn-n+1}\\
&= \min_{h\in Y^*}\frac{f_\delta(W(h))}{vn-n+1}=\frac{\min_{h\in Y^*}f(p(h))}{vn-n+1}.
\end{aligned}
\end{equation}

Next, we show that $\gamma(h)<\gamma(h')$ holds for all $h\in Y^*$ and all $h'\in Y_i\in\mathcal{Y}\setminus\{Y^*\}$. According to property \ref{item:StocP2} in \lemaref{StocPotential} and \eqref{Fdifferenceequ}, we have
\begin{equation}\begin{aligned}\label{StochasticInequ}
&\gamma(h)-\gamma(h')\\
&\leq (n+1)\sum_{Y_j\in\mathcal{Y}\setminus\{ Y^*\}}\min_{\hat{h}\in Y_j}f(p(\hat{h}))-\sum_{Y_j\in\mathcal{Y}\setminus\{ Y_i\}}\min_{\hat{h}\in Y_j}f(p(\hat{h}))\\
&=(n+1)\min_{\hat{h}\in Y_i}f(p(\hat{h}))-\min_{\hat{h}\in Y^*}f(p(\hat{h}))\\
&\quad+n\sum_{Y_j\in\mathcal{Y}\setminus\{ Y^*,Y_i\}}\min_{\hat{h}\in Y_j}f(p(\hat{h}))\\
&< \frac{n+1}{vn-n+1}\min_{\hat{h}\in Y^*}f(p(\hat{h}))-\min_{\hat{h}\in Y^*}f(p(\hat{h}))\\
&\quad + \frac{n(v-2)}{vn-n+1}\min_{\hat{h}\in Y^*}f(p(\hat{h}))=0.
\end{aligned}
\end{equation}
Thus, it follows from \eqref{StochasticInequ} and \lemaref{LimitDis} that $h$ is stochastically stable if and only if $h\in Y^*$.

Regarding \ref{item:MCMconc2}, it follows from conclusion \ref{item:MCMconc1} and \defiref{StableAbilityDefi} that as $\epsilon\to0$, $\pi_{h}^{\epsilon}>0$ if and only if $h\in Y^*$, implying that only $h$ in $Y^*$ can be visited with non-zero probability as $\epsilon\to0$. Since $Y^*$ is the induced RCC of $Q^*$, the conclusion directly follows from \eqref{Class2}.
\end{proof}

The memory subsystem may be required to be of very large size so that only
the joint strategies with maximum cycle-based metric can be observed with
non-zero probability. This requirement arises because $\bar{m}_{\delta}$ in
\eqref{boundmMwequ} is proportional to $L$ which can be very
large. Sometimes, the size of memory subsystem has a predefined length,
which may be less than $\bar{m}_{\delta}$. Thus, for this case, we
introduce a memory-based metric in \defiref{MBMeasure}.

Next, we propose a class of perturbations such that the joint strategies with the maximum memory-based metric can be observed frequently in the sense of stochastic stability. It is worth noting that \aspref{PayoffBound} is removed in this case, with respect to \thomref{ThomMCBM}, i.e., the bound of the payoff function is unnecessary. 

\begin{thom}[Maximum memory-based metric]\label{ThomMMBM}
	Let $G_2=(N,\mathcal{S},J)$ be an $n$-agent meta game on a finite joint strategy space $\mathcal{S}$. The training system has a memory of predefined length $m\in\mathbb{N}_{>0}$, and consider the memory-based metric $M_m(s)$ in \eqref{MemoryMeasureEqu} for each joint strategy $s\in\mathcal{S}$. Suppose that Assumption \ref{AspUmw} holds. If the following conditions are all true:
	\begin{enumerate}[label=(\Roman*)]
		\item\label{item:MMMcond1} there exists a $\delta_0>0$ such that $M_m(Q^*)\ge M_m(Q_i)+\delta_0$ for all sink equilibria $Q_i\in\mathcal{Q}\setminus\{Q^*\}$;
		\item\label{item:MMMcond2} $f$ is a $\delta$-feasible function for some $\delta\in(0,\delta_0)$;
	\end{enumerate} 
	then under the perturbed SBRD $P^{\epsilon}$, 
	\begin{enumerate}[label=(\roman*)]
		\item\label{item:MMMconc1} a history state $h$ is stochastically stable if and only if $h\in Y^*$;
		\item\label{item:MMMconc2} as $\epsilon\to 0$, a joint strategy $s\in\mathcal{S}$ is recorded in the memory subsystem with nonzero probability if and only if $s\in Q^*$,
	\end{enumerate}
	where $M_m(Q^*)=\max_{Q\in \mathcal{Q}}M_m(Q)$ and $Y^*$ is the induced RCC of $Q^*$.
\end{thom}
\begin{proof}
According to property \ref{item:CnectPM2} in \lemaref{ReCClass}, $W(Y_i)=M_m(Q_i)$ for all $Y_i\in\mathcal{Y}$. Thus, it follows from the conditions \ref{item:MMMcond1} and \ref{item:MMMcond2} that
\begin{equation}\begin{aligned}
W(Y^*)&=M_m(Q^*)\ge M_m(Q_i)+\delta_0\\
&\ge W(Y_i)+\delta_0>W(Y_i)+\delta,
\end{aligned}\end{equation}
for all RCCs $Y_i\in\mathcal{Y}\setminus\{Y^*\}$. By following the same argument \eqref{change}-\eqref{StochasticInequ} in the proof of \thomref{ThomMCBM}, we can obtain that $\gamma(h)<\gamma(h')$ for all $h\in Y^*$ and all $h'\in Y_i\in\mathcal{Y}\setminus\{Y^*\}$. Thus, the conclusions \ref{item:MMMconc1} and \ref{item:MMMconc2} follow from \lemaref{LimitDis} and the proof of \thomref{ThomMCBM}.
\end{proof}

\subsection{Approximation of the Maximum Metrics}
In this subsection, we remove the prior information about the lower bound
of the difference between the maximum and second maximum metrics in all
sink equilibria, i.e., $\delta_0$ is unknown. This relaxation is important
because the condition \ref{item:MCMcond1} in \thomref{ThomMCBM} and
condition \ref{item:MMMcond1} in \thomref{ThomMMBM} both need such a
$\delta_0$ which sometimes is hard to find. Conversely, we introduce a
design parameter $\bar{\delta}$ specifying the difference between the
metric of the sink equilibrium related to the stochastically stable states
and the maximum metric. Thus, the preassigned parameter $\bar{\delta}$
actually characterizes the stochastic stability of the joint strategies
observed with non-zero probability.

\begin{thom}[Approximation of maximum cycle-based metric]\label{ThomAMW}
	Let $G_2=(N,\mathcal{S},J)$ be an $n$-agent meta game on a finite joint strategy space $\mathcal{S}$. The training system has a memory of length $m\in\mathbb{N}_{>0}$, and consider the cycle-based metric $M_c(s)$ in \eqref{CycleMeasureEqu} for each joint strategy $s\in\mathcal{S}$. Suppose that Assumptions \ref{PayoffBound} and \ref{AspUmw} hold. For any given $\bar{\delta}>0$, if the following conditions are all true:
	\begin{enumerate}[label=(\Roman*)]
		\item $f$ is a $\delta$-feasible function for some $\delta\in(0,\bar{\delta})$;
		\item the memory $m$ is not less than $\bar{m}_\delta$ given by  
		\begin{equation}\label{boundAmMwequ}
		\bar{m}_{\delta}=\frac{2LJ_{\rm max}}{\bar{\delta}-\delta},
		\end{equation}
	\end{enumerate}
	then under the perturbed SBRD $P^\epsilon$, as $\epsilon\to 0$, every joint strategy $s\in\mathcal{S}$ which is recorded in the memory with nonzero probability, belongs to a sink equilibrium $Q_i$ satisfying
	\begin{equation}
	M_c(s)=M_c(Q_i)\ge M_c(Q^*)-\bar{\delta},
	\end{equation}
	where $M_c(Q^*)=\max_{Q\in \mathcal{Q}}M_c(Q)$.
\end{thom} 

\begin{proof}

Take a sink equilibrium $Q_j$ from $\mathcal{Q}$. Suppose that
\begin{equation}\label{ConverseEqu}
M_c(Q_j)<M_c(Q^*)-\bar{\delta},
\end{equation}
where $M_c(Q^*)=\max_{Q\in \mathcal{Q}}M_c(Q)$. We let $Y^*$ denote the induced RCC of $Q^*$.

Let $\bar{m}_\delta$ be as in \eqref{boundAmMwequ}. Then, if $m\ge\bar{m}_\delta$, it follows from the property \ref{item:CnectPM1} in \lemaref{ReCClass} that
\begin{equation}\label{RRCQequ2}
|W(Y_i)-M_c(Q_i)|<\frac{\bar{\delta}-\delta}{2},
\end{equation}
for all RCCs $Y_i\in\mathcal{Y}$. Then, by \eqref{ConverseEqu} and \eqref{RRCQequ2}, we have
\begin{equation*}\begin{aligned}
W(Y^*)&>M_c(Q^*)-\frac{\bar{\delta}-\delta}{2}\ge M_c(Q_j)+\bar{\delta}-\frac{\bar{\delta}-\delta}{2}\\
&=M_c(Q_j)+\frac{\bar{\delta}-\delta}{2}+\delta>W(Y_j)+\delta,
\end{aligned}
\end{equation*}
which is similar to \eqref{RCCdisequ}. By following the same arguments as in \eqref{change}-\eqref{StochasticInequ}, we have $\gamma(h)<\gamma(h')$ for all $h\in Y^*$ and all $h'\in Y_j\in\mathcal{Y}\setminus\{Y^*\}$. Thus, \lemaref{LimitDis} implies that for any $h\in Y_j$, $h$ is not stochastically stable.

If $h$ is stochastically stable, it follows from Theorem 4 in \cite{HPY:93} that $h$ is contained in an RCC $Y_i$. The sink equilibrium $Q_i$ related to $Y_i$ cannot satisfies \eqref{ConverseEqu}, because the above analysis shows that \eqref{ConverseEqu} will lead to a contradiction, i.e., $h$ is not stochastically stable. Thus, we have 
\begin{equation}\label{ConverseEqu2}
M_c(Q_i)\ge M_c(Q^*)-\bar{\delta}.
\end{equation}
The conclusion directly follows from the fact that as $\epsilon\to 0$, only the joint strategies contained in stochastic stable history states can be observed in the memory with nonzero probability. 
\end{proof}

Then, we have the following corollary.

\begin{cor}[Approximation of maximum memory-based metric]
	\label{ThomMMBMA}
	Let $G_2=(N,\mathcal{S},J)$ be an $n$-agent meta game on a finite joint strategy space $\mathcal{S}$. The training system has a memory of predefined length $m\in\mathbb{N}_{>0}$, and consider the memory-based metric $M_m(s)$ in \eqref{MemoryMeasureEqu} for each joint strategy $s\in\mathcal{S}$. Suppose that Assumption \ref{AspUmw} holds. For any given $\bar{\delta}>0$, if $f$ is a $\delta$-feasible function for some $\delta\in(0,\bar{\delta})$,
	then under the perturbed SBRD $P^\epsilon$, as $\epsilon\to 0$, every joint strategy $s\in\mathcal{S}$ which is recorded in the memory with nonzero probability, belongs to a sink equilibrium $Q_i$ satisfying
	\begin{equation}
	M_m(s)=M_m(Q_i)\ge M_m(Q^*)-\bar{\delta},
	\end{equation}
	where $M_m(Q^*)=\max_{Q\in \mathcal{Q}}M_m(Q)$.
\end{cor} 
\begin{proof}
This corollary is straightforward by combining the proofs of Theorems \ref{ThomMMBM} and \ref{ThomAMW}.
\end{proof}

\section{Conclusion}\label{sec:conclusion}
In this paper, we introduced two multi-agent policy evaluation metrics
based on the concept of sink equilibrium. The proposed framework is
practical and general in the sense that it can be easily applied to a large
class of stochastic meta-games. Compared with existing popular evaluation
methods, such as Nash equilibria, Elo ratings and $\alpha$-rankings, our
evaluation metrics can handle dynamical cyclical behaviors and are
compatible with single-agent reinforcement learning. In particular, the
memory-based metric is realistic settings with limitations on amount of
available memory.

We have shown that learning algorithms designed to seek coarse correlated
equilibria cannot be directly applied to compute sink equilibria with good
properties by demonstrating that these two concepts are not special cases
of each other. Then, we introduced a training system which consists of a
learning subsystem and a memory subsystem for multi-agent policy
seeking. Under some mild conditions, we proposed perturbed SBRD algorithms
which guarantee that only the joint strategies with the maximum underlying
metrics are observed with non-zero probability in the memory subsystem. We
also modified the training algorithms for learning joint strategies with
metrics close to the maximum by relaxing some conditions. Future work
involves applying the perturbed SBRD algorithms in practical scenarios.

\bibliographystyle{plainurl}
\bibliography{alias,references}

\end{document}